\pdfoutput=1
\newcommand*{\EXTENDED}{}%

\PassOptionsToPackage{table}{xcolor}
\documentclass[sigplan,screen, authorversion]{acmart}
\setcopyright{acmcopyright}
\acmConference[PLDI'18]{the 39th ACM SIGPLAN Conference on Programming Language Design and Implementation}{18--22 June 2018}{Philadelphia, PA, USA}



\usepackage{booktabs}
\usepackage{subcaption}
\usepackage{tikz}
\usetikzlibrary{matrix,fit,calc,positioning,backgrounds}
\usetikzlibrary{arrows,patterns,decorations}
\usetikzlibrary{decorations.pathreplacing,calligraphy}
\usetikzlibrary{tikzmark}

\usepackage{pgfplots}

\usepackage{enumitem}

\usepackage{framed}
\definecolor{shadecolor}{HTML}{EEEEEE}
\renewcommand\FrameSep{1.5mm}

\usepackage{balance}

\usepackage{array}
\newcolumntype{L}{>{$}l<{$}}
\newcolumntype{C}{>{$}c<{$}}
\newcolumntype{R}{>{$}r<{$}}
\newcolumntype{P}{>{$}p{\linewidth}<{$}}

\usepackage{multirow}

\newcommand\stack[2][l]{\begin{array}[t]{@{}#1@{}}#2\end{array}}

\catcode`"=\active
\def"#1"{\texttt{#1}}

\tikzset{event/.style={draw=none, inner sep=0.4mm}}
\newcommand\evtlbl[1]{\mbox{#1:~}}
\newcommand\evW[3]{W\def\temp{#1}\ifx\temp\empty\else,#1\fi\,{#2}}
\newcommand\evR[3]{R\def\temp{#1}\ifx\temp\empty\else,#1\fi\,{#2}}
\newcommand\evL{L}
\newcommand\evU{U}
\newcommand\evLt{L^{\rm t}}
\newcommand\evUt{U^{\rm t}}

\definecolor{colorco}{HTML}{3366CC}
\definecolor{colorrf}{HTML}{FF0000}
\definecolor{colorfr}{HTML}{33BB33}
\definecolor{colorhb}{HTML}{660000}
\definecolor{colortx}{HTML}{BB33BB}
\definecolor{colorlock}{HTML}{A37A11}
\definecolor{colorcomment}{HTML}{4AAD14}

\tikzset{
  edgeco/.style={colorco,-latex},
  edgerf/.style={colorrf,-latex},
  edgefr/.style={colorfr,dashed,-latex},
  edgepo/.style={black, -latex},
  edgepi/.style={black, dotted,->},
  stxn/.style={draw, colortx, inner sep=1mm, rounded corners},
  ftxn/.style={draw, dashed, colortx, inner sep=1mm, rounded corners},
  scri/.style={draw, double, colorlock, inner sep=1mm},
}

\newcommand\blacknum[2][0]{%
\def\myrowsep{0.35}%
\def\myradius{0.75em}%
\begin{tikzpicture}[baseline=-0.3em,yscale=-1]%
\draw[black!25, line width=\myradius, cap=round] (0,0) to (0,0);
\draw[overlay, black!25, line width=\myradius, cap=round] (0,0) to (0,\myrowsep*#1);
\node[inner sep=0pt, outer sep=0pt](a) at (0,0) {%
\makebox[0pt][c]{\textcolor{black}{\tiny\sf\bfseries #2}}
};%
\end{tikzpicture}%
}

\pgfdeclaredecoration{simple line}{initial}{
\state{initial}[width=\pgfdecoratedpathlength-1sp]{%
  \pgfmoveto{\pgfpointorigin}}
\state{final}{\pgflineto{\pgfpointorigin}}
}
\tikzset{
   shift left/.style={%
     decorate,decoration={simple line,raise=#1}},
   shift right/.style={%
     decorate,decoration={simple line,raise=-1*#1}},
}

\newcommand\dashboxed[2][0mm] {\begin{tikzpicture}[baseline=(a.base)]
\node[anchor=base, fill=colortx!15, inner sep=0.7mm, rounded corners](a)
{$\vphantom{(}#2\hspace*{#1}$};
\end{tikzpicture}\hspace*{-#1}}

\newcommand\modelcomment[1]{\text{\color{black!50}\itshape (#1)}}

\newcommand\Memalloy{\textrm{Memalloy}}
\newcommand\Litmus{\textrm{Litmus}}
\newcommand\MemSynth{\textrm{MemSynth}}
\newcommand\Diy{\textrm{Diy}}

\newcommand\cat{\texttt{.cat}}

\newcommand\ax[1]{\textsc{#1}}
\newcommand\ltest[1]{\textbf{#1}}

\newcommand\mhyphen{\textrm{-}}

\usepackage{pifont}
\newcommand\cmark{\ding{51}}%
\newcommand\xmark{\ding{55}}%
\newcommand\redtick{\textcolor{red!80!black}{\cmark}}
\newcommand\greencross{\textcolor{green!80!black}{\xmark}}
\usepackage{marvosym}
\newcommand\timeout{\Clocklogo}

\theoremstyle{acmdefinition}
\newtheorem{theorem}{Theorem}[section]
\newtheorem{remark}[theorem]{Remark}

\newenvironment{BoxedExample}[1][]{%
\begin{oframed}%
\vspace*{-1.5mm}%
\begin{example}[#1]%
}{%
\end{example}%
\vspace*{-1.5mm}%
\end{oframed}%
}

\newenvironment{Remark}[1][]{%
\begin{snugshade*}%
\begin{remark}[#1]%

}{%
\end{remark}%
\end{snugshade*}%
}

\begin{document}

\title{The Semantics of Transactions and Weak Memory \\ in x86, Power,
ARM, and C++}
\renewcommand\shorttitle{The Semantics of Transactions and Weak Memory}

\setcopyright{rightsretained}
\acmPrice{}
\acmDOI{10.1145/3192366.3192373}
\acmYear{2018}
\copyrightyear{2018}
\acmISBN{978-1-4503-5698-5/18/06}
\acmConference[PLDI'18]{39th ACM SIGPLAN Conference on Programming Language Design and Implementation}{June 18--22, 2018}{Philadelphia, PA, USA}

\setlength\leftmargini{\parindent}
\addtolength\leftmargini{2\labelsep}

\author{Nathan Chong}
\affiliation{\institution{ARM Ltd.} \country{United Kingdom}}
\author{Tyler Sorensen}
\affiliation{\institution{Imperial College London} \country{United Kingdom}}
\author{John Wickerson}
\affiliation{\institution{Imperial College London} \country{United Kingdom}}

\newcommand\TODO[1]{\textcolor{red}{#1}}
\newcommand\NCComment[1]{\textcolor{blue!80!black}{{\bf [[\ref{?}NC:} #1{\bf
]]}}}
\newcommand\TSComment[1]{\textcolor{red!70!black}{{\bf [[\ref{?}TS:} #1{\bf
]]}}}
\newcommand\JWComment[1]{\textcolor{green!70!black}{{\bf [[\ref{?}JW:} #1{\bf
]]}}}

\newcommand\totalNumberOfLitmusTests{\ref{?}}

\newcommand{\isabelleqed}{%
\begin{tikzpicture}[x=0.8mm, y=0.8mm, baseline=-0.3mm, line join=round]
\begin{scope}[yslant=-0.5]
  \draw (0,0) rectangle +(1,1);
  \draw (2,1) rectangle +(1,1);
  \draw (1,2) rectangle +(1,1);
\end{scope}
\begin{scope}[yslant=0.5]
  \filldraw (1,-1) rectangle +(1,1);
  \filldraw (3,-2) rectangle +(1,1);
  \filldraw (2,0) rectangle +(1,1);
\end{scope}
\begin{scope}[yslant=0.5,xslant=-1]
  \draw (1,0) rectangle +(1,1);
  \draw (2,-1) rectangle +(1,1);
  \draw (3,1) rectangle +(1,1);
\end{scope}
\end{tikzpicture}%
}

\begin{abstract}
Weak memory models provide a complex, system-centric semantics for concurrent programs, while transactional memory (TM) provides a simpler, programmer-centric semantics. 
Both have been studied in detail, but their \emph{combined} semantics is not well understood.  
This is problematic because such widely-used architectures and languages as x86, Power, and C++ all support TM, and all have weak memory models.

Our work aims to clarify the interplay between weak memory and TM by extending existing axiomatic weak memory models (x86, Power, ARMv8, and C++) with new rules for TM. 
Our formal models are backed by automated tooling that enables (1)~the synthesis of tests for validating our models against existing implementations and (2)~the model-checking of TM-related transformations, such as lock elision and compiling C++ transactions to hardware.
A key finding is that a proposed TM extension to ARMv8 currently being considered within ARM Research is incompatible with lock elision without sacrificing portability or performance.
\end{abstract}

\maketitle

\newcommand\po{\mathit{po}}
\newcommand\addr{\mathit{addr}}
\newcommand\ctrl{\mathit{ctrl}}
\newcommand\data{\mathit{data}}
\newcommand\rf{\mathit{rf}}
\newcommand\co{\mathit{co}}
\newcommand\fr{\mathit{fr}}
\newcommand\rmw{\mathit{rmw}}

\newcommand\Rel{\mathit{Rel}}
\newcommand\Acq{\mathit{Acq}}

\section{Introduction}

\emph{Transactional memory}~\cite{herlihy+93} (TM) is a concurrent programming abstraction that promises scalable performance without programmer pain. 
The programmer gathers instructions into \emph{transactions}, and the system guarantees that each appears to be performed entirely and instantaneously, or not at all.
To achieve this, a typical TM system tracks each transaction's memory accesses, and if it detects a conflict (i.e., another thread concurrently accessing the same location, at least one access being a write), resolves it by aborting the transaction and rolling back its changes.

\subsection{Motivating Example: Lock Elision in ARMv8}
\label{sec:intro:hlebug}

One important use-case of TM is \emph{lock elision}~\cite{rajwar+01, dice+09}, in which the lock/unlock methods of a mutex are skipped and the critical region (CR) is instead executed speculatively inside a transaction. 
If two CRs do not conflict, this method allows them to be executed simultaneously, rather than serially. 
If a conflict is detected, the transaction is rolled back and the system resorts to acquiring the mutex as usual.

Lock elision may not apply to all CRs, so an implementation must ensure mutual exclusion between transactional and non-transactional CRs.
This is typically done by starting each transactional CR with a read of the lock variable (and self-aborting if it is taken)~\cite[\S16.2.1]{intel17}.
If the mutex is subsequently acquired by a non-transactional CR then the TM system will detect a conflict on the lock variable and abort the transactional CR.

%
%

Thus, reasoning about lock elision requires a concurrency model that accounts for both modes, transactional and non-transactional. 
In particular, systems with memory models weaker than \emph{sequential consistency} (SC)~\cite{lamport79} must ensure that the non-transactional lock/unlock methods synchronise sufficiently with transactions to provide mutual exclusion.

In their seminal paper introducing lock elision, \citeauthor{rajwar+01} argued that ``correctness is guaranteed without any dependence on memory ordering''~\cite[\S9]{rajwar+01}.
In fact, by drawing on a decade of weak memory formalisations~\cite{alglave+14, flur+16, pulte+17} and by extending state-of-the-art tools~\cite{wickerson+17, alglave+11a, lustig+17}, we show it is straightforward to contradict this claim \emph{automatically}.
%
%
%
%
%
%
%

\begin{BoxedExample}[Lock elision is unsound under ARMv8]
\label{ex:hle}

\renewcommand\tabcolsep{0.4mm}
\renewcommand\arraystretch{0.9}

\newcommand\xrightbrace[2][1]{%
\def\mylineheight{0.35}%
\raisebox{2.1mm}{%
\smash{%
\begin{tikzpicture}[baseline=(top)]
\coordinate (top) at (0,#1*\mylineheight-0.06);
\coordinate (bottom) at (0,0);
\draw[colorcomment, pen colour={colorcomment}, decoration={calligraphic brace,amplitude=2.1pt}, decorate, line width=1pt] 
  (top) to node[auto, inner sep=0] {~~\begin{tabular}{l}#2\end{tabular}} (bottom);
\end{tikzpicture}}}}

\newcommand\lc[1]{\textcolor{colorlock}{#1}}

Consider the program below, in which two threads use CRs to update a shared location $x$.
\vspace*{1.5mm}
\begin{center}
\small
\begin{tabular}{@{}ll@{\hspace{1mm}}||@{\hspace{1mm}}ll@{}}
\hline
\multicolumn{4}{c}{Initially: $"[X0]"=x=0$}                     \\
\hline
\lc{"lock()"}       &            & \lc{"lock()"}     &          \\
    "LDR W5,[X0]"   & \xrightbrace[3]{$x \leftarrow x+2$} 
                                 &     "MOV W7,\#1"  & 
                              \xrightbrace[2]{$x \leftarrow 1$} \\
    "ADD W5,W5,\#2" &            &     "STR W7,[X0]" &          \\
    "STR W5,[X0]"   &            & \lc{"unlock()"}   &          \\
\lc{"unlock()"}     &                                           \\
\hline
\multicolumn{4}{c}{Test: $x=2$}                                 \\
\hline
\end{tabular}
\end{center}
It must not terminate with $x=2$, for this would violate mutual exclusion.
Now, let us instantiate the lock/unlock calls with two possible implementations of those methods.
\begin{center}
\vspace*{-1.5mm}
\small
\begin{tabular}{@{}lll@{\hspace{1mm}}||@{\hspace{1mm}}lll@{}}
\hline
\multicolumn{6}{c}{Initially: $"[X0]"=x=0$, $"[X1]"=m=0$}    \\
\hline
\blacknum[2]{1}& \lc{"Loop:"}        & 
 \xrightbrace[6]{atomically \\ update $m$ \\ from 0 \\ to 1} & 
\blacknum[7]{3}& \lc{"TXBEGIN"}   & 
 \xrightbrace{begin txn}                                     \\
               & \lc{"LDAXR W2,[X1]"}&                       & 
               & \lc{"LDR W6,[X1]"}  & 
 \xrightbrace[4]{load $m$ \\ and abort \\ if non-\\ zero}        \\
               & \lc{"CBNZ W2,Loop"} &                       & 
               & \lc{"CBZ W6,L1"}    &                       \\
\blacknum[2]{4}& \lc{"MOV W3,\#1"}   &                       & 
               & \lc{"TXABORT"}      &                       \\
               & \lc{"STXR W4,W3,[X1]"} &                    & 
               & \lc{"L1:"}          &                       \\
               & \lc{"CBNZ W4,Loop"} &                       & 
               &    "MOV W7,\#1"     & 
 \xrightbrace[2]{$x \leftarrow 1$}                           \\
\blacknum{2}   &    "LDR W5,[X0]"    & 
 \xrightbrace[3]{$x \leftarrow x+2$}                         & 
               &    "STR W7,[X0]"    &                       \\
\blacknum[2]{5}&    "ADD W5,W5,\#2"  &                       & 
               &\lc{"TXEND"}         & 
\xrightbrace{end txn}                                        \\
               &    "STR W5,[X0]"    &                       & 
               &                     &                       \\
               &\lc{"STLR WZR,[X1]"} & 
\xrightbrace{$m \leftarrow 0$}                               &      
                                     &                       \\
\hline
\multicolumn{6}{c}{Test: $x=2$}                              \\
\hline
\end{tabular}
\end{center}
\vspace*{1.5mm}
The left thread executes its CR non-transactionally, using the recommended ARMv8 spinlock~\cite[K9.3]{arm17}, while the right thread uses lock elision (with unofficial but representative TM instructions).
This program \emph{can} terminate with $x=2$, thus witnessing the unsoundness of lock elision, as follows:

{
\renewcommand\labelenumi{\blacknum{\theenumi}}
\begin{enumerate}[leftmargin=*]
\item The left thread reads the lock variable $m$ as $0$ (free). 
"LDAXR" indicates an \emph{acquire} load, which means that the read cannot be reordered with any later event in program-order.

\item The left thread reads $x$ as $0$. 
This load can execute speculatively because ARMv8 does not require that the earlier store-exclusive ("STXR") completes first~\cite{pulte+17}.

\item The right thread starts a transaction, sees the lock is still free, updates $x$ to $1$, and commits its transaction.

\item The left thread updates $m$ to $1$ (taken). 
This is a store-exclusive ("STXR")~\cite{jensen+87}, so it only succeeds if $m$ has not been updated since the last load-exclusive ("LDAXR"). 
It does succeed, because the right thread only \emph{reads} $m$.

\item Finally, the left thread updates $x$ to $2$, and $m$ to $0$. "STLR" is a \emph{release} store, which means that the update to $m$ cannot be reordered with any earlier event in program-order.
\end{enumerate}
}

\end{BoxedExample}

The crux of our counterexample is that a (non-transaction\-al) CR can start executing after the lock has been observed to be free, but before it has actually been taken. 
Importantly, this relaxation is safe if all CRs are mutex-protected (i.e., the spinlock \emph{in isolation} is correct), since every lock acquisition involves writing to the lock variable and at most one store-exclusive can succeed. 
Rather, the counterexample only arises when this relaxation is \emph{combined} with any reasonable TM extension to ARMv8.
This includes a proposed extension currently being considered within ARM Research.

Furthermore, there appears to be no easy fix.
Re-implemen\-ting the spinlock by appending a "DMB" fence to the "lock()" implementation would prevent the problematic reordering, but would also inhibit compatibility with code that uses the ARM-recommended spinlock, and may decrease performance when lock elision is not in use.
Otherwise, if software portability is essential, transactional CRs could be made to \emph{write} to the lock variable (rather than just read it), but this would induce serialisation, and thus nullify the potential speedup from lock elision.

\subsection{Our Work}

In this paper, we use formalisation to tame the interaction between TM and weak memory.
Specifically, we propose axiomatic models for how transactions behave in x86~\cite{intel17}, Power~\cite{power30}, ARMv8~\cite{arm17}, and C++~\cite{c++tm15}.
As well as the lock elision issue already explained, our formalisations revealed:
\begin{itemize}
\item an ambiguity in the specification of Power TM (\S\ref{sec:x86_power:adding_txns}), 
\item a bug in a register-transfer level (RTL) prototype implementation of ARMv8 TM (\S\ref{sec:arm:testing}), 
\item a simplification to the C++ TM proposal (\S\ref{sec:cpp:adding_transactions}), and 
\item that coalescing transactions is unsound in Power (\S\ref{sec:metatheory:monotonicity}).
\end{itemize}

Although TM is conceptually simple, it is notoriously challenging to implement correctly, as exemplified by Intel repeatedly having to disable TM in its processors due to bugs~\cite{hachman14, skl105}, IBM describing adding TM to Power as ``arguably the single-most invasive change ever made to IBM's RISC architecture''~\cite{adir+14}, and the C++ TM Study Group listing ``conflict with the C++ memory model and atomics'' as one of their hardest challenges~\cite{wong14}.
To cope with the combined complexity of transactions and weak memory that exist in real systems, we build on several recent advances in automated tooling to help develop and validate our models. 
In the x86 and Power cases, we use the SAT-based \Memalloy{} tool \cite{wickerson+17}, extended with an exhaustive enumeration mode \`a la \citet{lustig+17}, to automatically synthesise exactly the `minimally forbidden' tests (up to a bounded size) that distinguish our TM models from their respective non-TM baselines.
We then use the \Litmus{} tool~\cite{alglave+11a} to check that these tests are never observed on existing hardware (i.e., that our models are sound).
We also generate a set of `maximally allowed' tests, which we use to assess the completeness of our models (i.e., how many of the behaviours our models allow are empirically observable).

Moreover, we investigate several properties of our models.
For instance, C++ offers `atomic' transactions and `relaxed' transactions; we prove that atomic transactions are strongly isolated, and that race-free programs with no non-SC atomics and no relaxed transactions enjoy `transactional SC'. 
Other properties of our models we verify up to a bound using \Memalloy{}; these are that introducing, enlarging, or coalescing transactions introduces no new behaviours, and that C++ transactions compile soundly to x86, Power, and ARMv8.

Finally, we show how \Memalloy{} can be used to check a library
implementation against its specification by encoding it as a program
transformation.
We apply our technique to check that x86 and Power lock elision libraries correctly implement mutual exclusion -- but that this is not so, as we have seen, in ARMv8.

\paragraph{Summary} Our contributions are as follows:

\begin{itemize}

\item a fully-automated toolflow for generating tests from an axiomatic memory model and using them to validate the model's soundness, its completeness, and its metatheoretical properties (\S\ref{sec:methodology});

\item formalisations of TM in the SC (\S\ref{sec:transactions}), x86 (\S\ref{sec:x86_power}), Power (\S\ref{sec:x86_power}), ARMv8 (\S\ref{sec:arm}), and C++ (\S\ref{sec:cpp}) memory models;

\item proofs that the transactional C++ memory model guarantees strong isolation for atomic transactions, and transactional SC for race-free programs with no non-SC atomics or non-atomic transactions (\S\ref{sec:cpp}); 

\item the automatic, bounded verification of transactional monotonicity and compilation from C++ transactions to hardware (\S\ref{sec:metatheory}); and

\item a technique for validating lock elision against hardware TM models, which is shown to be effective through the discovery of the serious flaw of Example~\ref{ex:hle} (\S\ref{sec:metatheory}).

\end{itemize}

\paragraph{Companion Material} We provide all the models we propose (in the \cat{} format~\cite{alglave+14}), the automatically-generated litmus tests used to validate our models, litmus tests corresponding to all the executions discussed in our paper, and Isabelle proofs of all statements marked with the \isabelleqed{} symbol.

\section{Background: Axiomatic Memory Models}
\label{sec:memory_models}

\newcommand\semi{\mathbin{\hspace{-0.2ex};\hspace{-0.2ex}}}
\newcommand\eqdef{=}
\newcommand\id{\mathit{id}}
\newcommand\domain{\mathsf{domain}}
\newcommand\range{\mathsf{range}}
\newcommand\imm{\mathsf{imm}}
\renewcommand\min{\mathsf{min}}
\newcommand\acyclic{\mathbf{acyclic}}
\newcommand\irreflexive{\mathbf{irreflexive}}
\newcommand\isempty{\mathbf{empty}}
\newcommand\EXT[1]{#1{_{\mathrm{e}}}}
\newcommand\INT[1]{#1{_{\mathrm{i}}}}
\newcommand\LOC[1]{#1{_{\mathrm{loc}}}}
\newcommand\DLOC[1]{#1{_{\neq\mathrm{loc}}}}
\newcommand\Exec{\mathbb{X}}

\newcommand\loc{\mathit{loc}}
\newcommand\sloc{\mathit{sloc}}

\newcommand\rfe{\EXT{\rf\hspace*{-1.4pt}}}
\newcommand\coe{\EXT{\co}}
\newcommand\fre{\EXT{\fr\hspace*{-0.5pt}}}
\newcommand\rfi{\INT{\rf\hspace*{-1.4pt}}}
\newcommand\coi{\INT{\co\hspace{0.5pt}}}
\newcommand\fri{\INT{\fr}}

\newcommand\com{\mathit{com}}
\newcommand\come{\EXT{\com}}
\newcommand\comi{\INT{\com}}

\newcommand\hb{\mathit{hb}}

Here we give the necessary background on the formal framework we use for reasoning about programs, which is standard across several recent works ~\cite{alglave+14, wickerson+17, lustig+17}.

A \emph{memory model} defines how threads interact with shared memory. 
An \emph{axiomatic} memory model consists of constraints (i.e., axioms) on \emph{candidate executions}.
An execution is a graph representing a runtime behaviour, whose structure is defined below.
The candidate executions of a program are obtained by assuming a non-deterministic memory system: each load can observe a store from anywhere in the program.
After filtering away the candidates that fail the constraints, we are left with the \emph{consistent} executions; i.e., those that are allowed in the presence of the actual memory system.

\subsection{Executions}
\label{sec:memory_models:executions}
Let $\Exec$ be the set of all executions.
Each execution is a graph whose vertices, $E$, represent runtime memory-related events and whose labelled edges represent various relations between them. 
The events are partitioned into $R$, $W$, and $F$, the sets of read, write, and fence events.\footnote{
We encode fences as \emph{events} (rather than edges) because this simplifies execution minimisation (\S\ref{sec:methodology:empirical}). 
We then derive architecture-specific fence relations that connect events separated by fence events, which we use in our models and execution graphs.} 
Events in an execution are connected by the following relations:
\begin{itemize}

\item $\po$, program order (a.k.a.~sequenced-before);

\item $\addr$/$\ctrl$/$\data$, an address/control/data dependency;

\item $\rmw$, to indicate read-modify-write operations;

\item $\sloc$, between events that access the same location;

\item $\rf$, the `reads-from' relation; and

\item $\co$, the `coherence' order in which writes hit memory.

\end{itemize}
We restrict our attention to executions that are \emph{well-formed} as follows: 
$\po$ forms, for each thread, a strict total order over that thread's events; 
$\addr$, $\ctrl$, and $\data$ are within $\po$ and always originate at a read;
$\rmw$ links the read of an RMW operation to its corresponding write;
$\rf$ connects writes to reads accessing the same location, with no read having more than one incoming $\rf$ edge; and
$\co$ connects writes to the same location and forms, for each location, a strict total order over the writes to that location.

\paragraph{Notation} 
Given a relation $r$, $r^{-1}$ is its inverse, $r^?$ is its reflexive closure, $r^+$ is its transitive closure, and $r^*$ is its reflexive transitive closure. 
We use $\neg$ for the complement of a set or relation, implicitly with respect to the set of all events or event pairs in the execution. 
We write `${;}$' for relational composition: $r_1\semi r_2 = \{(x,z)\mid \exists y\ldotp (x,y)\in r_1 \wedge (y,z)\in r_2\}$. 
We write $[-]$ to lift a set to a relation: $[s] = \{(x,x)\mid x\in s\}$. 
To restrict a relation $r$ to being inter-thread or intra-thread, we use $\EXT{r} = r\setminus(\po \cup \po^{-1})^*$ or $\INT{r} = r\cap (\po \cup \po^{-1})^*$, respectively. 
Similarly, $\LOC{r} = r\cap\sloc$. 

\paragraph{Derived Relations} 
The \emph{from-read} ($\fr$) relation relates each read event to all the write events on the same location that are $\co$-later than the write the read observed~\cite{lustig+17}. The $\com$ relation captures three ways events can `communicate' with each other.
\begin{eqnarray*}
\fr &=& ([R]\semi\sloc\semi[W]) \setminus (\rf^{-1}\semi (\co^{-1})^*) \\
\com &=& \rf \cup \co \cup \fr
\end{eqnarray*}

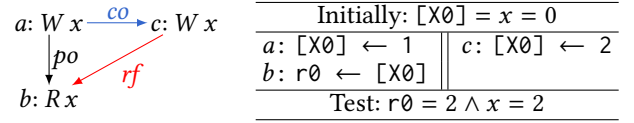
\begin{figure}
\centering
\begin{tikzpicture}[inner sep=1pt,baseline=4mm]
\node (a1) at (0,1) {$\evtlbl{$a$}\evW{}{x}{1}$};
\node (a2) at (0,0) {$\evtlbl{$b$}\evR{}{x}{3}$};
\node (b1) at (1.8,1) {$\evtlbl{$c$}\evW{}{x}{1}$};
\draw[edgeco] (a1) to [auto] node {$\co$} (b1);
\draw[edgepo] (a1) to [auto] node {$\po$} (a2);
\draw[edgerf] (b1) to [auto] node {$\rf$} (a2);
\end{tikzpicture}
\hspace*{2mm}
\renewcommand\arraystretch{0.9}
\begin{tabular}{@{~}r@{~}l||r@{~}l@{~}}
\hline
\multicolumn{4}{c}{Initially: $"[X0]"=x=0$}                      \\ 
\hline
"$a$:" & "[X0] $\leftarrow$ 1"  & "$c$:" & "[X0] $\leftarrow$ 2" \\
"$b$:" & "r0 $\leftarrow$ [X0]" &                                \\ 
\hline
\multicolumn{4}{c}{Test: $"r0"=2 \wedge x=2$}                    \\ 
\hline
\end{tabular}
\caption{An execution and its litmus test}
\label{fig:sample}
\end{figure}

\paragraph{Visualising Executions} 
We represent executions using diagrams like the one in Fig.~\ref{fig:sample} (left).
Here, the $\po$-edges separate the execution's events into two threads, each drawn in one column. 
Each event is labelled with the sets it belongs to, such as $R$ and $W$.
We use location names such as $x$ to identify the $\sloc$-classes.

\subsection{From Executions to Litmus Tests}
\label{sec:memory_models:litmus}

In order to test whether an execution of interest is observable in practice, it is necessary to convert it into a \emph{litmus test} (i.e., a program with a postcondition)~\cite{collier92}. 
This litmus test is constructed so that the postcondition only passes when the particular execution of interest has been taken~\cite{alglave+10, wickerson+17}.

As an example, the execution on the left of Fig.~\ref{fig:sample} corresponds to the pseudocode litmus test on the right. 
Read events become loads, writes become stores, and the $\po$-edges induce the order of instructions and their partitioning into threads. 
To ensure that the litmus test passes only when the intended $\rf$-edges are present, we arrange that each store writes a unique non-zero value, and then check that each local register holds the value written by the store it was intended to observe -- this corresponds to the $"r0" = 2$ in the postcondition.
To ensure that the intended $\co$-edges are present, we check the final value of each memory location -- this corresponds to the $x=2$ in the postcondition.\footnote{
When there are more than two writes to a location, extra constraints on executions are needed to fix all the $\co$-edges~\cite{wickerson+17}.}

\section{Axiomatising Transactions}
\label{sec:transactions}

\newcommand\stxn{\mathit{stxn}}
\newcommand\stronglift{\mathsf{stronglift}}
\newcommand\weaklift{\mathsf{weaklift}}

\newcommand\Order{\ax{Order}}
\newcommand\TxnOrder{\ax{TxnOrder}}

\newcommand\axlabel[1]{\textsc{(#1)}}
\newcommand\axiom[2]{\multicolumn{3}{@{}P}{\hspace*{1.2mm}#2\hfill\axlabel{#1}\hspace*{0.7mm}}}
\newcommand\where{\quad\text{where}~}
\newcommand\header[1]{\multicolumn{3}{@{}>{\raggedright}p{\linewidth}}{#1}}
\newlength{\myframesep}
\setlength{\myframesep}{3pt} 
\newlength{\axwidth}
\setlength{\axwidth}{\linewidth}
\addtolength{\axwidth}{-1.3mm}
\addtolength{\axwidth}{-2\myframesep}

\newcommand\newaxiom[2]{\multicolumn{3}{@{}P}{\dashboxed[\axwidth]{}#2 \hfill \axlabel{#1}\hspace*{0.7mm}}}

\newenvironment{axiomatisationWithoutBox}{%
\renewcommand\arraystretch{1.1}%
\begin{tabular*}{\linewidth}{@{}R@{~~}C@{~~}L@{}}
}{
\end{tabular*}
}

\newenvironment{axiomatisation}{%

\renewcommand\FrameSep{\myframesep}%
\begin{framed}%
\begin{axiomatisationWithoutBox}%
}{
\end{axiomatisationWithoutBox}%
\end{framed}%
}

Transactional memory (TM) can be provided either at the architecture level (x86, Power, ARMv8) or in software (C++). 
Since we are concerned only with the \emph{specification} of TM, and not its implementation, we can formalise both forms within a unified framework.
In this section, we describe how program executions can be extended to express transactions (\S\ref{sec:transactions:executions}) and how we can derive litmus tests to test for these executions (\S\ref{sec:transactions:litmus}). 
We then propose axioms for capturing the \emph{isolation} of transactions (\S\ref{sec:transactions:isolation}), and for strengthening the SC memory model to obtain \emph{transactional} SC (\S\ref{sec:transactions:tsc}).

\subsection{Transactional Executions}
\label{sec:transactions:executions}

To enable transactions in an axiomatic memory modelling framework, we extend executions with an $\stxn$ relation that relates events in the same successful (i.e., committed) transaction.
For an execution to be well-formed, $\stxn$ must be a partial equivalence relation (i.e., symmetric and transitive), and each $\stxn$-class must coincide with a contiguous subset of $\po$.
When generating the candidate executions of a program with transactions, each transaction is assumed to succeed or fail non-deterministically. 
That is, each either gives rise to a $\stxn$-class of events, or vanishes as a no-op.

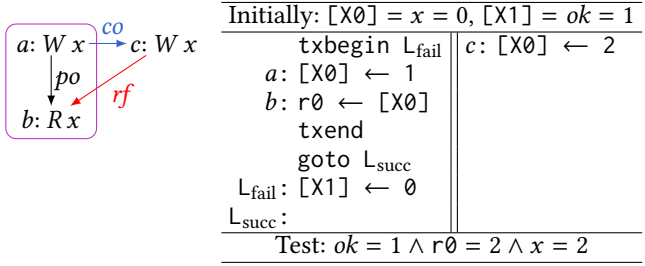
\begin{figure}
\centering
\begin{tikzpicture}[inner sep=1pt]
\node (a1) at (0,1) {$\evtlbl{$a$}\evW{}{x}{1}$};
\node (a2) at (0,0) {$\evtlbl{$b$}\evR{}{x}{3}$};
\node (b1) at (1.5,1) {$\evtlbl{$c$}\evW{}{x}{1}$};
\draw[edgeco] (a1) to [auto,pos=0.6] node {$\vphantom{p}\co$} (b1);
\draw[edgepo] (a1) to [auto] node {$\po$} (a2);
\draw[edgerf] (b1) to [auto] node {$\rf$} (a2);
\node[stxn, fit=(a1)(a2)] {};
\end{tikzpicture}
\hfill
\renewcommand\arraystretch{0.9}
\renewcommand\tabcolsep{1mm}
\begin{tabular}{r@{~}l||r@{~}l}
\hline
\multicolumn{4}{c}{Initially: $"[X0]"=x=0$, $"[X1]"=\mathit{ok}=1$} \\ 
\hline
       & "txbegin L$_{\rm fail}$" & $c$":" & "[X0] $\leftarrow$ 2"  \\
$a$":" & "[X0] $\leftarrow$ 1"    &                                 \\
$b$":" & "r0 $\leftarrow$ [X0]"   &                                 \\
       & "txend"                  &                                 \\
       & "goto L$_{\rm succ}$"    &                                 \\  
"L$_{\rm fail}$:" & "[X1] $\leftarrow$ 0" &                         \\  
"L$_{\rm succ}$:" &               &                                 \\ 
\hline
\multicolumn{4}{c}{Test: $\mathit{ok}=1 \wedge "r0"=2 \wedge x=2$}  \\ 
\hline
\end{tabular}
\caption{A transactional execution and its litmus test }
\label{fig:sample_txn}
\end{figure}

Diagrammatically, we represent $\stxn$ using \smash{%
\begin{tikzpicture}[baseline=(a.base)]
\node[inner sep=0](a){boxes.};
\node[stxn, fit=(a)] {};
\end{tikzpicture}}
For instance, events $a$ and $b$ in Fig.~\ref{fig:sample_txn} form a successful transaction.


%
%

\begin{Remark}
To study the behaviour of unsuccessful transactions in more detail, one could add an explicit representation of them in executions, perhaps using 
\smash{%
\begin{tikzpicture}[baseline=(a.base)]
\node[inner sep=0](a){dashed boxes.};
\node[ftxn, fit=(a)] {};
\end{tikzpicture}}
However, the behaviour of unsuccessful transactions is tricky to ascertain on hardware because of the rollback mechanism.
Moreover, it is unclear how they should interact with $\co$, since $\co$ is the order in which writes hit the memory, which writes in unsuccessful transactions never do.
\end{Remark}

\subsection{From Transactional Executions to Litmus Tests}
\label{sec:transactions:litmus}

A transactional execution can be converted into a litmus test by extending the construction of \S\ref{sec:memory_models:litmus}.
As an example, the execution on the left of Fig.~\ref{fig:sample_txn} corresponds to the litmus test on the right.
The instructions in the transaction simply need enclosing in instructions that begin and end a transaction. 
We write these as "txbegin" and "txend" here; our tooling specialises these for each target architecture. 
The postcondition checks that the transaction succeeded using the `$\mathit{ok}$' location, which is zeroed in the transaction's fail-handler, "L$_{\rm fail}$", the label of which is provided with the "txbegin" instruction.

\subsection{Weak and Strong Isolation}
\label{sec:transactions:isolation}

We now explain how the \emph{isolation} property of transactions can be captured as a property of an execution graph.
%
A TM system provides \emph{weak} isolation if transactions are isolated from other transactions; that is, their intermediate state cannot affect or be affected by other transactions~\cite{blundell+06, harris+10}. 
It provides \emph{strong} isolation if transactions are also isolated from non-transactional code.

\begin{figure}
\begin{subfigure}[b]{0.24\linewidth}
\begin{tikzpicture}[inner sep=1pt, baseline=0.5cm]
\node (a1) at (0,0.7) {$\evR{}{x}{0}$};
\node (a2) at (0,0) {$\evR{}{x}{1}$};
\node (b1) at (1.3,0.35) {$\evW{}{x}{1}$};
\draw[edgefr] (a1) to [auto] node {$\fr$} (b1);
\draw[edgerf] (b1) to [auto] node {$\rf$} (a2);
\draw[edgepo] (a1) to [auto] node (po) {$~\po$} (a2);
\node[stxn, fit=(a1)(a2)(po)] {};
\end{tikzpicture}
\caption{}
\end{subfigure}
\begin{subfigure}[b]{0.24\linewidth}
\begin{tikzpicture}[inner sep=1pt, baseline=0.5cm]
\node (a1) at (0,0.7) {$\evR{}{x}{0}$};
\node (a2) at (0,0) {$\evW{}{x}{2}$};
\node (b1) at (1.3,0.35) {$\evW{}{x}{1}$};
\draw[edgefr] (a1) to [auto] node {$\fr$} (b1);
\draw[edgeco] (b1) to [auto] node {$\co$} (a2);
\draw[edgepo] (a1) to [auto] node (po) {$~\po$} (a2);
\node[stxn, fit=(a1)(a2)(po)] {};
\end{tikzpicture}
\caption{}
\end{subfigure}
\begin{subfigure}[b]{0.24\linewidth}
\begin{tikzpicture}[inner sep=1pt, baseline=0.5cm]
\node (a1) at (0,0.7) {$\evW{}{x}{1}$};
\node (a2) at (0,0) {$\evR{}{x}{2}$};
\node (b1) at (1.3,0.35) {$\evW{}{x}{2}$};
\draw[edgeco] (a1) to [auto] node {$\co$} (b1);
\draw[edgerf] (b1) to [auto] node {$\rf$} (a2);
\draw[edgepo] (a1) to [auto] node (po) {$~\po$} (a2);
\node[stxn, fit=(a1)(a2)(po)] {};
\end{tikzpicture}
\caption{}
\end{subfigure}
\begin{subfigure}[b]{0.24\linewidth}
\begin{tikzpicture}[inner sep=1pt, baseline=0.5cm]
\node (a1) at (0,0.7) {$\evW{}{x}{1}$};
\node (a2) at (0,0) {$\evW{}{x}{2}$};
\node (b1) at (1.3,0.35) {$\evR{}{x}{1}$};
\draw[edgerf] (a1) to [auto] node {$\rf$} (b1);
\draw[edgefr] (b1) to [auto] node {$\fr$} (a2);
\draw[edgeco, shift right=0.5mm] (a1) to [auto,swap] node (co) {$\co~$} (a2);
\draw[edgepo, shift left=0.5mm] (a1) to [auto] node (po) {$~\po$} (a2);
\node[stxn, fit=(a1)(a2)(po)(co)] {};
\end{tikzpicture}
\caption{}
\end{subfigure}
\caption{Four SC executions that are allowed by weak isolation but forbidden by strong isolation}
\label{fig:weak_strong_isolation}
\end{figure}
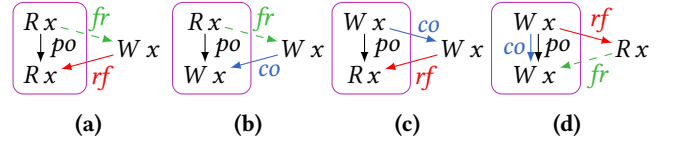

\newcommand\AtomicRMW{\ax{RMWIsol}}

The four 3-event SC executions in Fig.~\ref{fig:weak_strong_isolation} illustrate the difference between strong and weak isolation.
In each, the interfering event would need to be within a transaction to be forbidden by weak isolation; strong isolation does not make this distinction.
Executions {\bf (a)} and {\bf (d)} correspond to what \citeauthor{blundell+06} call \emph{non-interference} and \emph{containment}, respectively, and {\bf (b)} is similar to the standard axiom for RMW isolation (cf. \AtomicRMW{} in Fig.~\ref{fig:axioms_x86}).

\newcommand\WeakIsol{\ax{WeakIsol}}
\newcommand\StrongIsol{\ax{StrongIsol}}

Failures of isolation can be characterised as communication cycles between transactions. 
To define these cycles, the following constructions are useful:
\begin{eqnarray*}
\weaklift(r,t) &=& t\semi (r\setminus t)\semi t \\
\stronglift(r,t) &=& t^?\semi (r\setminus t)\semi t^?.
\end{eqnarray*}
If $r$ relates events $e_1$ and $e_2$ in different transactions, then $\weaklift(r,\stxn)$ relates all the events in $e_1$'s transaction to all those in $e_2$'s transaction. 
The $\stronglift$ version also includes edges where the source and/or the target event are not in a transaction.
Weak and strong isolation can then be axiomatised by treating all the events in a transaction as a single event whenever the transaction communicates with another transaction (\WeakIsol) or any other event (\StrongIsol).
\begin{align}
\tag{\WeakIsol}
\acyclic(\weaklift(\com,\stxn)) \\
\tag{\StrongIsol}
\acyclic(\stronglift(\com,\stxn))
\end{align}
%

\subsection{Transactional Sequential Consistency}
\label{sec:transactions:tsc}

\begin{figure}
\centering
\begin{axiomatisation}
\axiom{\Order}{\acyclic(\hb) \where \hb = \po\cup\com}
\\
\newaxiom{\TxnOrder}{\acyclic(\stronglift(\hb,\stxn))} \\
\end{axiomatisation}
\caption{SC axioms~\cite{shasha+88}, with TSC extensions  \dashboxed{\text{highlighted}}}
\label{fig:axioms_tsc}
\end{figure}

Although isolation is a critical property for transactions, it only provides a lower bound on the guarantees that real architectures provide.
Meanwhile, an upper bound on the guarantees provided by a reasonable TM implementation is \emph{transactional sequential consistency} (TSC)~\cite{dalessandro+09}. 
The models we propose in \S\ref{sec:x86_power}--\ref{sec:cpp} all lie between these bounds.

TSC is a strengthening of the SC memory model in which consecutive events in a transaction must appear consecutively in the overall execution order.
Where SC can be characterised axiomatically (Fig.~\ref{fig:axioms_tsc}) by forbidding cycles in program order and communication (\Order{})~\cite{shasha+88}, we can obtain TSC by additionally forbidding such cycles between transactions and non-transactional events (\TxnOrder{}).
Note that \TxnOrder{} subsumes the \StrongIsol{} axiom.



\section{Methodology}
\label{sec:methodology}

\newcommand\conformance{\mathsf{min\mhyphen{}inconsistent}}
\newcommand\maxconsistent{\mathsf{max\mhyphen{}consistent}}
\newcommand\consistent{\mathsf{consistent}}

\newcommand\mfence{\mathit{mfence}}
\newcommand\lwsync{\mathit{lwsync}}
\newcommand\sync{\mathit{sync}}
\newcommand\REL{\mathit{Rel}}
\newcommand\SC{\mathit{SC}}

We identify three components of a memory modelling methodology: (1) developing and refining axioms, (2) synthesising and running conformance tests, and (3) checking metatheoretical properties. 
In this section, we explain our approach to each of these components, and in particular, how we have extended the \Memalloy{} tool~\cite{wickerson+17} to support each task.

\paragraph{Background on \Memalloy{}}

The original \Memalloy{} tool, built on top of Alloy~\cite{jackson12a}, was developed for comparing memory models.
It takes two models (say, $M$ and $N$), and searches for a single execution that distinguishes them (i.e., is inconsistent under $M$ but consistent under $N$). 
Additionally, if \Memalloy{} is supplied with a translation on executions (e.g., representing a compiler mapping or a compiler optimisation), then it searches for a witness that the translation is unsound.
This translation is defined by a relation, typically named $\pi$, from `source' events to `target' events.

\subsection{Developing and Refining Axioms}

For each model, we make a first attempt at a set of axioms using information obtained from specifications, programming manuals, research papers, and discussions with designers.
Then, for each proposed change to the model, we use \Memalloy{} to generate tests that become disallowed or allowed as a result.
We decide whether to accept the change based on discussing these tests with designers, and running them on existing hardware (where available) using the \Litmus{} tool~\cite{alglave+11a}.

In order to extend \Memalloy{} to support the development of transactional memory models in this way, we augmented the form of executions as described in \S\ref{sec:transactions:executions}, and modified the litmus test generator as described in \S\ref{sec:transactions:litmus}.

\subsection{Synthesising and Running Conformance Tests}
\label{sec:methodology:empirical}

To build confidence in a model, we compare the behaviours it admits against those allowed by the architecture or language being modelled. 
It is vital that no behaviour allowed by the architecture/language is forbidden by the model, so we exhaustively generate all litmus tests (up to a bounded size) that our model forbids, and confirm using \Litmus{} that none can be observed on existing hardware.

To achieve this, we extended \Memalloy{} with a mode for exhaustively generating conformance tests for a given model $M$. 
The key to exhaustive generation is a suitable notion of \emph{minimality}, without which we would obtain an infeasibly large number of tests.
We closely follow \citet{lustig+17}, and define execution minimality with respect to the following partial order between executions. Let $X\sqsubset Y$ hold when execution $X$ can be obtained from execution $Y$ by:
\begin{enumerate}[label=(\roman*)]

\item removing an event (plus any incident edges),

\item removing a dependency edge ($\addr$, $\ctrl$, $\data$, $\rmw$), or

\item downgrading an event (e.g. reducing an acquire-read to a plain read in ARMv8).

\end{enumerate}
We then calculate the set $\conformance(M) \eqdef \{X \in \Exec \mid X \notin \consistent(M) \wedge \forall X'\sqsubset X\ldotp X'\in\consistent(M)\}$.


Extending \Memalloy{} to support the synthesis of transactional conformance tests requires minimality to take transactions into account.
To do this, we arrange that $X \sqsubset Y$ also holds when $X$ can be obtained from $Y$ by:
\begin{enumerate}[label=(\roman*), start=5]
\item making the first or last event in a transaction non-trans\-actional (i.e. removing all of its incident $\stxn$ edges). 
\end{enumerate}
(We avoid the `middle' of a transaction so as not to create non-contiguous transactions and hence ill-formed executions.)
\begin{Remark}
While this is a slightly coarse notion of minimality -- a more refined version would also allow a large transaction to be chopped into two smaller ones -- it is cheap to implement in the constraint solver as it only requires quantification over a single event.
As a result, \Memalloy{} may generate some executions that appear non-minimal, but as we show in~\S\ref{sec:x86_power:testing}, this does not impede our ability to generate and run large batches of conformance tests.
\end{Remark}

\paragraph{Generating Allowed Tests}

Having generated the minimal\-ly-forbidden tests, the question naturally arises of whether we can generate the \emph{maximally-allowed} tests too.
Where the minimally-forbidden tests include just enough fences/de\-pen\-den\-cies/transactions to be forbidden (and failing to observe these tests empirically suggests that the model is not too strong), the maximally-allowed tests include just \emph{not} enough (and observing them suggests that the model is not too weak).
We found the maximally-allowed tests valuable for communicating with engineers about the detailed relaxations permitted by our models.
However, in our experiments, allowed tests are less conclusive than forbidden ones, because where the observation of a forbidden test implies that the model is unsound, the non-observation of an allowed test may be caused by not performing enough runs, or by the machine under test being implemented conservatively.

Moreover, the notion of execution maximality is not as natural as minimality. 
For instance, an inconsistent execution is only considered minimally-inconsistent if \emph{removing} any event makes it \emph{consistent}, yet it is not sensible to deem a consistent execution maximally-consistent only when \emph{adding} any event makes it \emph{inconsistent} -- such a condition is almost impossible to satisfy. 
Even with event addition/removal set aside, maximal-consistency tends to require executions to be littered with redundant fences and dependencies.

Therefore, we approximate the maximally-consistent executions as those obtained via a single $\sqsubset$-step from a minimally-inconsistent execution. That is, we let $\maxconsistent(M) \eqdef \{X \in \Exec \mid \exists Y\in\conformance(M)\ldotp X \sqsubset Y\}$.








\subsection{Checking Metatheoretical Properties}
\label{sec:methodology:libabs}

As explained at the start of this section, \Memalloy{} is able to validate transformations between two memory models, providing they can be encoded as a $\pi$-relation between executions. 
In \S\ref{sec:metatheory}, we exploit this ability to check several TM-related transformations and compiler mappings.

In fact, \Memalloy{} can also be used to check libraries under weak memory.
Prior work has (manually) verified that stack, queue, and barrier libraries implement their specifications under weak memory models~\cite{batty+13, sorensen+16a}; here we show how checking these types of properties can be automated up to a bounded number of library and client events.
We see this as a straightforward first-step towards a general verification effort.
The idea, which we apply to a lock elision library in \S\ref{sec:metatheory:hle}, is to treat the replacement of the library's specification with its implementation as a program transformation.
To do this, we first extend executions with events that represent method calls.
Second, we extend execution well-formedness so that illegal call sequences (such as popping from an empty stack) are rejected.
Third, we strengthen the memory model's consistency predicate with axioms capturing the library's obligations (such as pops never returning data from later pushes).
Finally, we constrain $\pi$ so that it maps each method call to an event sequence representing the implementation of that method.

\newcommand\Locked{\mathit{L}}

\newcommand\ppo{\mathit{ppo}}
\newcommand\implied{\mathit{implied}}
\newcommand\tfence{\mathit{tfence}}

\newcommand\Coherence{\ax{Coherence}}

\newcommand\isync{\mathit{isync}}

\newcommand\fence{\mathit{fence}}
\newcommand\cfence{\mathit{efence}}
\newcommand\propI{\mathit{prop_1}}
\newcommand\propII{\mathit{prop_2}}
\newcommand\tpropI{\mathit{tprop_1}}
\newcommand\tpropII{\mathit{tprop_2}}
\newcommand\prop{\mathit{prop}}
\newcommand\ihb{\mathit{ihb}}
\newcommand\thb{\mathit{thb}}
\newcommand\XW{\mathit{XW}}

\newcommand\Propagation{\ax{Propagation}}
\newcommand\Observation{\ax{Observation}}
\newcommand\ExclWrites{\ax{ExclWrites}}
\newcommand\TxnCancelsRMW{\ax{TxnCancelsRMW}}

\section{Transactions in x86 and Power}
\label{sec:x86_power}

Over the next three sections, we show how our methodology can be applied to four different targets.
We begin with x86 and Power, which have both supported TM since 2013~\cite{intel12, cain+13}. 
Intel's Transactional Synchronisation Extensions (TSX) provide "XBEGIN", "XEND", and "XABORT" instructions for starting, committing, and aborting transactions, while Power provides "tbegin", "tend", and "tabort".

%
%

\subsection{Background: the x86 and Power Memory Models}
Both the x86 memory model~\cite{owens+09} and the Power memory model~\cite{alglave+14, sarkar+11, sarkar+12} allow certain instructions to execute out of program order, with x86 allowing stores to be reordered with later loads and Power allowing many more relaxations.
Both architectures provide fences ("MFENCE" in x86, and "lwsync", "sync", and "isync" in Power) to allow these relaxations to be controlled.
The x86 architecture provides atomic RMWs via "LOCK"-prefixed instructions, while Power implements RMWs using \emph{exclusive} instructions like those seen in Example~\ref{ex:hle}.
Moreover, x86 is \emph{multicopy-atomic}~\cite{collier92}, which means that writes are propagated to all other threads simultaneously.
Power does not have this property, so its memory model includes explicit axioms to describe how writes propagate.

More formally, we extend executions with relations that connect events in program order that are separated by a fence event of a given type.
For x86, we add an $\mfence$ relation, and for Power, we add $\isync$, $\lwsync$, and $\sync$.
%

\begin{figure}
\centering
\begin{axiomatisation}
\axiom{\Coherence}{\acyclic(\LOC{\po} \cup \com)}
\\
\axiom{\AtomicRMW}{\isempty(\rmw \cap (\fre \semi \coe))} \\
\axiom{\Order}{\acyclic(\hb)}
\\
\where \ppo &=& \stack{((W \times W) \cup (R\times W) \cup (R\times R)) \cap \po}
\\
\dashboxed[63mm]{}\tfence &=& \po\cap((\neg\stxn\semi\stxn) \cup (\stxn\semi\neg\stxn))
\\
\Locked &=& \domain(\rmw) \cup \range(\rmw) \\
\implied &=& \stack{[\Locked]\semi\po \cup \po\semi[\Locked] \dashboxed{{}\cup \tfence}}
\\
\hb &=& \mfence \cup \ppo \cup \implied \cup \rfe \cup \fr \cup \co
\\
\newaxiom{\StrongIsol}{\acyclic(\stronglift(\com,\stxn))} \\
\newaxiom{\TxnOrder}{\acyclic(\stronglift(\hb,\stxn))} \\
\end{axiomatisation}
\caption{x86 consistency axioms~\cite{alglave+14}, with our
TM additions \dashboxed{\text{highlighted}}}
\label{fig:axioms_x86}
\end{figure}

An x86 execution is consistent if it satisfies all of the axioms in Fig.~\ref{fig:axioms_x86} (ignoring the highlighted regions for now).
The \Coherence{} axiom forbids cycles in communication edges and program order among events on the same location; this guarantees programs that use only a single location to have SC semantics.
Happens-before ($\hb$) imposes the event-ordering constraints upon which all threads must agree, and \Order{} ensures that $\hb^*$ is a partial order.
The constraints on $\hb$ arise from: fences placed by the programmer ($\mfence$), fences created implicitly by "LOCK"'d operations ($\implied$), the preserved fragment of the program order ($\ppo$), inter-thread observations ($\rfe$) and communication edges ($\fr$ and $\co$).

\begin{figure}
\centering
\begin{axiomatisation}
\axiom{\Coherence}{\acyclic(\LOC{\po} \cup \com)}
\\
\axiom{\AtomicRMW}{\isempty(\rmw \cap (\fre\semi \coe))}
\\
\axiom{\Order}{\acyclic(\hb)}
\\
\where \ppo &=& \modelcomment{preserved program order, elided}
\\
\dashboxed[63mm]{}\tfence &=& \po\cap((\neg\stxn\semi\stxn) \cup (\stxn\semi\neg\stxn))
\\
\fence &=& \sync \dashboxed{{}\cup \tfence} \cup (\lwsync \setminus (W\times R))
\\
\ihb &=& \ppo \cup \fence
\\
\dashboxed[65mm]{}\thb &=& (\rfe \!\cup ((\fre \!\cup\! \coe)^*\!\semi\ihb))^*\!\semi(\fre \!\cup\! \coe)^*\!\semi\rfe^?
\\
\hb &=& (\rfe^?\semi\ihb\semi \rfe^?) \dashboxed{{}\cup \weaklift(\thb,\stxn)}
\\
\axiom{\Propagation}{\acyclic(\co \cup \prop)}
\\
\where \cfence &=& \rfe^?\semi\fence\semi \rfe^?
\\
\propI &=& [W]\semi \cfence\semi\hb^*\semi[W]
\\
\propII &=& \come^*\!\semi \cfence^*\!\semi\hb^*\!\semi (\sync \dashboxed{{}\cup\tfence}) \semi \hb^*
\\
\dashboxed[32mm]{}\tpropI &=& \rfe\semi\stxn\semi[W]
\\
\dashboxed[25mm]{}\tpropII &=& \stxn\semi\rfe
\\
\prop &=& \propI \cup \propII \dashboxed{{}\cup \tpropI \cup \tpropII}
\\
\axiom{\Observation}{\irreflexive(\fre\semi \prop \semi \hb^*)}
\\
\newaxiom{\StrongIsol}{\acyclic(\stronglift(\com,\stxn))} \\
\newaxiom{\TxnOrder}{\acyclic(\stronglift(\hb,\stxn))} \\
\newaxiom{\TxnCancelsRMW}{\isempty(\rmw \cap \tfence^*)} \\
\end{axiomatisation}
\caption{Power consistency axioms~\cite{alglave+14}, with our
TM additions \dashboxed{\text{highlighted}}, and some details elided for brevity.}
\label{fig:axioms_power}
\end{figure}

A Power execution is consistent if it satisfies all the axioms in Fig.~\ref{fig:axioms_power} (again, ignoring the highlights).
The first axiom not already seen is \Order{}, which ensures that happens-before ($\hb$) is acyclic.
In contrast to x86, the happens-before relation in Power is formed from inter-thread observations ($\rfe$), the preserved fragment of the program order ($\ppo$), and fences ($\fence$). 
We elide the definition of $\ppo$ as it is complex and unchanged by our TM additions.
The $\prop$ relation governs how fences restrict ``the order in which writes propagate''~\cite{alglave+14}, and the \Propagation{} axiom ensures that this relation does not contradict the coherence order. 
\Observation{} governs which writes a read can observe: if $e_1$ propagates before $e_2$, then any read that happens after $e_2$ is prohibited from observing a write that precedes $e_1$ in coherence order. 

\subsection{Adding Transactions}
\label{sec:x86_power:adding_txns}

To extend the x86 and Power memory models to support TM, we make the following amendments, each highlighted in Figs.~\ref{fig:axioms_x86} and~\ref{fig:axioms_power}.

\paragraph{Strong Isolation (x86 and Power)}
The Power manual says that transactions ``appear atomic with respect to both transactional and non-transactional accesses performed by other threads''~\cite[\S5.1]{power30}, and the TSX manual defines conflicts not just between transactions, but between a transaction and ``another logical processor'' (which is not required to be executing a transaction)~\cite[\S16.2]{intel17}.
We interpret these statements to mean that x86 and Power transactions provide \emph{strong} isolation, so we add our \StrongIsol{} axiom from \S\ref{sec:transactions:isolation}.
%
%

\paragraph{Implicit Transaction Fences (x86 and Power)}
In both x86 and Power, fences are created at the boundaries of successful transactions. In x86, ``a successfully committed [transaction] has the same ordering semantics as a "LOCK" prefixed instruction''~\cite[\S16.3.6]{intel17}, and in Power, ``[a] "tbegin" instruction that begins a successful transaction creates a [cumulative] memory barrier'', as does ``a "tend" instruction that ends a successful transaction''~\cite[\S1.8]{power30}. 
Hence, we define $\tfence$ as the program-order edges that enter ($\neg\stxn\semi\stxn$) or exit ($\stxn\semi\neg\stxn$) a successful transaction, and add $\tfence$ alongside the existing fence relations ($\mfence$ and $\sync$).



\paragraph{Transaction Atomicity (x86 and Power)}
We extend the prohibition on $\hb$ cycles among events to include cycles among transactions (\TxnOrder). 
This essentially treats all the transaction's events as one indivisible event, and is justified by the atomicity guarantee given to transactions, which in x86 is ``that all memory operations [\ldots] appear to have occurred instantaneously when viewed from other logical processors''~\cite[\S16.2]{intel17}, and in Power is that each successful transaction ``appears to execute as an atomic unit as viewed by other processors and mechanisms''~\cite[\S1.8]{power30}.

\paragraph{Barriers within Transactions (Power only)}
Each transaction contains an ``integrated memory barrier'', which ensures that writes observed by a successful transaction are propagated to other threads before writes performed by the transaction itself~\cite[\S1.8]{power30}. 
This behaviour is epitomised by the \ltest{WRC}-style execution below~\cite[Fig.~6]{cain+13},
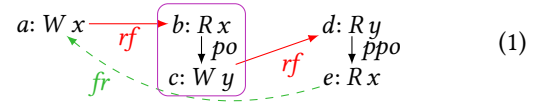
\begin{equation}
\label{eq:cain_wrc}
\begin{tikzpicture}[inner sep=1pt, baseline=0.35cm]
\node (a1) at (0,0.7) {\evtlbl{$a$}$\evW{}{x}{1}$};
\node (b1) at (2,0.7) {\evtlbl{$b$}$\evR{}{x}{1}$};
\node (b2) at (2,0) {\evtlbl{$c$}$\evW{}{\smash{y}}{1}$};
\node (c1) at (4,0.7) {\evtlbl{$d$}$\evR{}{\smash{y}}{1}$};
\node (c2) at (4,0) {\evtlbl{$e$}$\evR{}{x}{0}$};

\draw[edgefr, overlay] (c2) to [auto, bend left=28, pos=0.8] node {$\fr$} (a1);
\draw[edgerf] (a1) to [auto,swap] node {$\rf$} (b1);
\draw[edgerf] (b2) to [auto,swap] node {$\rf$} (c1);
\draw[edgepo] (b1) to [auto] node {$~\po$} (b2);
\draw[edgepo] (c1) to [auto] node {$~\ppo$} (c2);
\node[stxn, fit=(b1)(b2)] {};
\end{tikzpicture}
\end{equation}
which must be ruled out because the transaction's write ($c$) has propagated to the third thread before a write ($a$) that the transaction observed.
We capture this constraint by extending the $\prop$ relation so that it connects any write observed by a transaction to any write within that transaction ($\tpropI$).
In execution~\eqref{eq:cain_wrc}, this puts a $\prop$ edge from $a$ to $c$.
The execution is thus forbidden by the existing \Observation{} axiom.

\begin{Remark}\label{remark:rwc}
The following executions are similar to \eqref{eq:cain_wrc}, and like \eqref{eq:cain_wrc}, they could not be observed empirically. 
However, the Power manual is ambiguous about whether they should be forbidden.
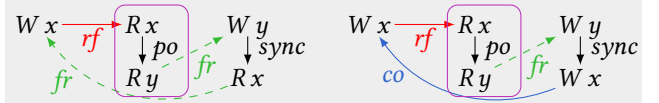
\begin{equation*}
\begin{tikzpicture}[inner sep=1pt, baseline=0.5cm]
\node (a1) at (0,0.7) {$\evW{}{x}{1}$};
\node (b1) at (1.4,0.7) {$\evR{}{x}{1}$};
\node (b2) at (1.4,0) {$\evR{}{\smash{y}}{0}$};
\node (c1) at (2.8,0.7) {$\evW{}{\smash{y}}{1}$};
\node (c2) at (2.8,0) {$\evR{}{x}{0}$};

\draw[edgefr, overlay] (c2) to [auto, bend left=38, pos=0.8] node {$\fr$} (a1);
\draw[edgerf] (a1) to [auto,swap] node {$\rf$} (b1);
\draw[edgefr] (b2) to [auto,swap] node {$\fr$} (c1);
\draw[edgepo] (b1) to [auto] node (po) {$~\po$} (b2);
\draw[edgepo] (c1) to [auto] node {$~\sync$} (c2);
\node[stxn, fit=(b1)(b2)(po)] {};
\end{tikzpicture}
\hspace*{5mm}
\begin{tikzpicture}[inner sep=1pt, baseline=0.5cm]
\node (a1) at (0,0.7) {$\evW{}{x}{2}$};
\node (b1) at (1.4,0.7) {$\evR{}{x}{2}$};
\node (b2) at (1.4,0) {$\evR{}{\smash{y}}{0}$};
\node (c1) at (2.8,0.7) {$\evW{}{\smash{y}}{1}$};
\node (c2) at (2.8,0) {$\evW{}{x}{1}$};

\draw[edgeco, overlay] (c2) to [auto, bend left=38, pos=0.8] node {$\co$} (a1);
\draw[edgerf] (a1) to [auto,swap] node {$\rf$} (b1);
\draw[edgefr] (b2) to [auto,swap] node {$\fr$} (c1);
\draw[edgepo] (b1) to [auto] node (po) {$~\po$} (b2);
\draw[edgepo] (c1) to [auto] node {$~\sync$} (c2);
\node[stxn, fit=(b1)(b2)(po)] {};
\end{tikzpicture}
\end{equation*}
In particular, because the transactions are read-only, we cannot appeal to the integrated memory barrier.
We have reported this ambiguity to IBM architects, and while we await a clarified specification, our model errs on the side of caution by permitting these executions.
\end{Remark}

\paragraph{Propagation of Transactional Writes (Power only)}
Although Power is not multicopy-atomic in general, \emph{transactional} writes are multicopy-atomic; that is, the architecture will ``propagate the transactional stores fully before committing the transaction''~\cite[\S4.2]{cain+13}. 
This behaviour is epitomised by another \ltest{WRC}-style execution, in which the middle thread sees the transactional write to $x$ before the right thread does.
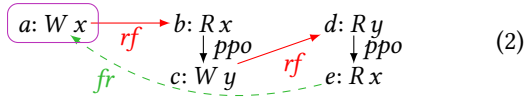
\begin{equation}
\label{eq:mca_wrc}
\begin{tikzpicture}[inner sep=1pt, baseline=0.35cm]
\node (a1) at (0,0.7) {\evtlbl{$a$}$\evW{}{x}{1}$};
\node (b1) at (2,0.7) {\evtlbl{$b$}$\evR{}{x}{1}$};
\node (b2) at (2,0) {\evtlbl{$c$}$\evW{}{\smash{y}}{1}$};
\node (c1) at (4,0.7) {\evtlbl{$d$}$\evR{}{\smash{y}}{1}$};
\node (c2) at (4,0) {\evtlbl{$e$}$\evR{}{x}{0}$};

\draw[edgefr,overlay] (c2) to [auto, bend left=25, pos=0.8] node {$\fr$} (a1);
\draw[edgerf] (a1) to [auto,swap] node {$\rf$} (b1);
\draw[edgerf] (b2) to [auto,swap] node {$\rf$} (c1);
\draw[edgepo] (b1) to [auto] node {$~\ppo$} (b2);
\draw[edgepo] (c1) to [auto] node {$~\ppo$} (c2);
\node[stxn, fit=(a1)] {};
\end{tikzpicture}
\end{equation}
To rule out such executions, it suffices to extend the $\prop$ relation with reads-from edges that exit a transaction ($\tpropII$), and then to invoke \Observation{} again.

\paragraph{Read-modify-writes (Power only)}
In Power, when a store-exclusive is separated from its corresponding load-exclusive by ``a state change from Transactional to Non-transactional or Non-transactional to Transactional'', the RMW operation will always fail~\cite[\S1.8]{power30}.
Therefore, the \TxnCancelsRMW{} axiom ensures that no consistent execution has an $\rmw$ edge crossing a transaction boundary.

\paragraph{Transaction Ordering (Power only)}
The Power manual states that ``successful transactions are serialised in some order'', and that it is impossible for contradictions to this order to be observed~\cite[p.~824]{power30}.

We capture this constraint by extending the $\hb$ relation to include a new $\thb$ relation between transactions. The $\thb$ relation imposes constraints on the order in which transactions can be serialised.
By including it in $\hb$ and requiring $\thb$ to be a partial order, we guarantee the existence of a suitable transaction serialisation order, without having to construct this order explicitly.

The definition of the $\thb$ relation is a little convoluted, but the intuition is quite straightforward: it contains all non-empty chains of intra-thread happens-before edges ($\ihb$) and inter-thread communication edges ($\come$), except those that contain an $\fre$ or $\coe$ edge followed by an $\rfe$ edge that does not terminate the chain.
The rationale for excluding $\fre\semi\rfe$ and $\coe\semi\rfe$ chains is that these do not provide ordering in a non-multicopy-atomic architecture.
That is, from
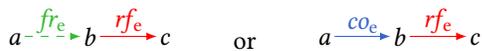
\begin{center}
\begin{tikzpicture}[inner sep=1pt]
\node (a) at (0,0) {\vphantom{A}$a$};
\node (b) at (1,0) {\vphantom{A}$b$};
\node (c) at (2,0) {\vphantom{A}$c$};
\draw[edgefr] (a) to [auto] node {$\fre$} (b);
\draw[edgerf] (b) to [auto] node {$\rfe$} (c);
\end{tikzpicture}
\qquad or \qquad
\begin{tikzpicture}[inner sep=1pt]
\node (a) at (0,0) {\vphantom{A}$a$};
\node (b) at (1,0) {\vphantom{A}$b$};
\node (c) at (2,0) {\vphantom{A}$c$};
\draw[edgeco] (a) to [auto] node {$\coe$} (b);
\draw[edgerf] (b) to [auto] node {$\rfe$} (c);
\end{tikzpicture}
\end{center}
we cannot deduce that $a$ happens before $c$, because this behaviour can also be attributed to the write $b$ being propagated to $c$'s thread before $a$'s thread.

\citeauthor{cain+13} epitomise the transaction-ordering constraint using the \ltest{IRIW}-style execution reproduced below~\cite[Fig.~5]{cain+13}.
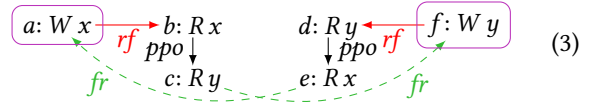
\begin{equation}
\label{eq:cain_iriw}
\begin{tikzpicture}[inner sep=1pt, baseline=0.35cm]
\node (a1) at (0,0.7) {\evtlbl{$a$}$\evW{}{x}{1}$};
\node (b1) at (1.8,0.7) {\evtlbl{$b$}$\evR{}{x}{1}$};
\node (b2) at (1.8,0) {\evtlbl{$c$}$\evR{}{\smash{y}}{0}$};
\node (c1) at (3.6,0.7) {\evtlbl{$d$}$\evR{}{\smash{y}}{1}$};
\node (c2) at (3.6,0) {\evtlbl{$e$}$\evR{}{x}{0}$};
\node (d1) at (5.4,0.7) {\evtlbl{$f$}$\evW{}{\smash{y}}{1}$};

\draw[edgefr, overlay] (c2) to [auto, bend left=30, pos=0.8] node {$\fr$} (a1);
\draw[edgefr, overlay] (b2) to [auto,swap, bend right=30, pos=0.8] node {$\fr$} (d1);
\draw[edgerf] (a1) to [auto,swap] node {$\rf$} (b1);
\draw[edgerf] (d1) to [auto] node {$\rf$} (c1);
\draw[edgepo] (c1) to [auto] node {$~\ppo$} (c2);
\draw[edgepo] (b1) to [auto,swap] node {$\ppo~$} (b2);
\node[stxn, fit=(a1)] {};
\node[stxn, fit=(d1)] {};
\end{tikzpicture}
\end{equation}
The execution must be disallowed because different threads observe incompatible transaction orders: the second thread observes $a$ before $f$, but the third observes $f$ before $a$.
Our model disallows this execution on the basis of a $\thb$ cycle between the two transactions.

We must be careful not to overgeneralise here, because a behaviour similar to \eqref{eq:cain_iriw} but with only \emph{one} write transactional was observed during our empirical testing, and is duly allowed by our model.

\subsection{Empirical Testing}
\label{sec:x86_power:testing}

\begin{table}
\caption{Testing our transactional x86 and Power models}
\small
\renewcommand\tabcolsep{3.9pt}
\begin{tabular}{@{}lrrrrrrrr@{}}
\toprule
\textbf{Arch.} & \textbf{$\left|E\right|$} & \raisebox{-1.45ex}{\smash{\begin{tabular}{r@{}}\bf Synthesis \\ \bf time (s) \end{tabular}}} & \multicolumn{3}{c}{\textbf{Forbid}} & \multicolumn{3}{c}{\textbf{Allow}} \\[-4pt]
\cmidrule(l){4-6}\cmidrule(l){7-9}
& & & \phantom{00}T & \phantom{00}S & \phantom{0}$\neg$S & \phantom{00}T & \phantom{00}S & \phantom{0}$\neg$S \\
\midrule
x86   &  2 &       4 &    0 &    0 &    0 &    2 &    2 &    0\\
      &  3 &      22 &    4 &    0 &    4 &   24 &   23 &    1\\
      &  4 &      87 &   22 &    0 &   22 &   99 &   99 &    0\\
      &  5 &     260 &   42 &    0 &   42 &  249 &  244 &    5\\
      &  6 &    4402 &  133 &    0 &  133 &  895 &  832 &   63\\
      &  7 & $>$7200 &  307 &    0 &  307 & 2457 & 1901 &  556\\
\midrule
\multicolumn{3}{r}{\textbf{Total (x86):}} & 508 &    0 &  508 & 3726 & 3101 &  625\\
\midrule
Power &  2 &      13 &    2 &    0 &    2 &    7 &    7 &    0\\
      &  3 &      58 &    9 &    0 &    9 &   44 &   44 &    0\\
      &  4 &     318 &   60 &    0 &   60 &  184 &  175 &    9\\
      &  5 &    9552 &  353 &    0 &  353 & 1517 & 1330 &  187\\
      &  6 & $>$7200 &  922 &    0 &  922 & 5043 & 4407 &  636\\
\midrule
\multicolumn{3}{r}{\textbf{Total (Power):}} & 1346 &    0 & 1346 & 6795 & 5963 &  832\\
\bottomrule
\end{tabular}
\label{table:x86_power}
\end{table}

%
%
%
%

Table~\ref{table:x86_power} gives the results obtained using our testing strategy from~\S\ref{sec:methodology:empirical}.
We use \Memalloy{} to synthesise litmus tests that are forbidden by our transactional models but allowed under the non-transactional baselines (the {\bf Forbid} set), up to a bounded number of events ($\left|E\right|$). 
We then derive the {\bf Allow} sets by relaxing each test.
We report synthesis times on a 4-core Haswell i7-4771 machine with 32GB RAM, using a timeout of 2 hours.
For both sets we give the number of tests (T) found; we say this number is complete if synthesis did not reach timeout and non-exhaustive otherwise.
We say a test is seen (S) if it is observed on any implementation, and not seen ($\neg$S) otherwise.
Each x86 test is run 1M times on four TSX implementations:
a Haswell (i7-4771),
a Broadwell-Mobile (i7-5650U),
a Skylake (i7-6700),
and a Kabylake (i7-7600U).
Each Power test is run 10M times on an 80-core POWER8 (TN71-BP012) machine. 
When testing this machine, we use \Litmus{}'s \emph{affinity} parameter~\cite{alglave+11a}, which places threads incrementally across the logical processors to encourage \ltest{IRIW}-style behaviours.

We were able to generate the complete set of x86 {\bf Forbid} executions that have up to 6 events, and the complete set of Power {\bf Forbid} executions up to 5 events.
Regarding these bounds: we remark that since our events only represent memory accesses and fences (not, for instance, starting or committing transactions), we can capture many interesting behaviours with relatively few events. 
For instance, these bounds are large enough to include all the executions discussed in this section.

Of our 508 x86 {\bf Forbid} tests, 29\% had one transaction, 44\% had two, and 27\% had three, 
and of the 1346 Power {\bf Forbid} tests, 29\% had one transaction, 54\% had two, and 17\% had three. 
No {\bf Forbid} test was empirically observable on either architecture, which gives us confidence that our models are not too strong.
Of the x86 {\bf Allow} tests, 83\% could be observed on at least one implementation, as could 88\% of the Power {\bf Allow} tests; this provides some evidence that our models are not excessively weak.
Many of the unobserved Power {\bf Allow} tests are based on the load-buffering (\ltest{LB}) shape, which has never actually been observed on a Power machine, even without transactions~\cite{alglave+14d}.

\def\Xtotalhours{34}
\def\Xtotaltests{313}

\begin{figure}
\centering
\begin{tikzpicture}
\begin{axis}[
height=30mm,
width=80mm,
ymin=-15,
ymax=100,
xmin=-1,
xmax=\Xtotalhours,
extra x ticks = {\Xtotalhours},
every tick/.style={line width=0.4pt},
xlabel={Time (hours)},
ylabel={\rotatebox{-90}{\begin{tabular}{@{}l@{}}Tests \\ found \\ (\%)\end{tabular}}},
axis line style={opacity=0},
tick pos=left,
tick align=outside,
clip=false,
]
\addplot[
mark=no,
draw=black,
]
table[x expr={\thisrow{secs}/3600}, y expr={\thisrow{solns}/\Xtotaltests*100}] {
secs   solns
 0       0
60      0
 60      0
120     0
 120     27
240     27
 240     74
480     74
 480     159
960     159
 960     212
1920    212
 1920    272
3840    272
 3840    295
7680    295
 7680    307
15360   307 
 15360   310
30720   310
 30720   313
122909  313
};
\draw ({rel axis cs:0,0} -|{axis cs:0,0})
-- ({rel axis cs:0,0} -|{axis cs:\Xtotalhours,0});
\draw ({rel axis cs:0,0} |-{axis cs:0,0})
-- ({rel axis cs:0,0} |-{axis cs:0,100});
\end{axis}
\end{tikzpicture}
\caption{The distribution of synthesis times for the 7-event x86 {\bf Forbid} tests}
\label{fig:exec_histogram}
\end{figure}

Increasing the timeout to 48 hours is sufficient to generate the complete set of x86 {\bf Forbid} executions for 7 events.
It takes \Xtotalhours{} hours for \Memalloy{} to find all \Xtotaltests{} tests.
Figure~\ref{fig:exec_histogram} shows how the percentage of executions found is affected by various caps on the synthesis time.
We observe that many tests are found quickly: 98\% of the tests are found within 2 hours (i.e., 6\% of the total synthesis time), and all of the tests are found within 9 hours (the remaining synthesis time is used to confirm that there are no further tests).
During the development process, we exploited this observation to obtain preliminary test results more rapidly.

\section{Transactions in ARMv8}
\label{sec:arm}

\newcommand\dmbld{\mathit{dmbld}}
\newcommand\dmbst{\mathit{dmbst}}
\newcommand\dmb{\mathit{dmb}}
\newcommand\isb{\mathit{isb}}

\newcommand\ob{\mathit{ob}}
\newcommand\aob{\mathit{aob}}
\newcommand\bob{\mathit{bob}}
\newcommand\dob{\mathit{dob}}

\newcommand\notxn{\mathit{notxn}}
\newcommand\poextend{\mathit{poextend}}


The ARMv8 memory model sits roughly between x86 and Power.
Like x86, it is multicopy-atomic~\cite{pulte+17}, but like Power, it permits several relaxations to the program order. 
Unwanted relaxations can be inhibited either using barriers ("DMB", "DMB\,LD", "DMB\,ST", "ISB") or using \emph{release}/\emph{acquire} instructions ("LDAR", "STLR") that act like one-way fences.

Formally, ARMv8 executions are obtained by adding six extra fields: $\Acq$ and $\Rel$, which are the sets of acquire and release events, and $\dmb$/$\dmbld$/$\dmbst$/$\isb$, which relate events in program order that are separated by barriers.

\begin{figure}
\centering
\begin{axiomatisation}
\axiom{\Coherence}{\acyclic (\LOC{\po} \cup \com)}
\\
\axiom{\Order}{\acyclic(\ob)}
\\
\where \dob &=& \modelcomment{order imposed by dependencies, elided}
\\
\aob &=& \modelcomment{order imposed by atomic RMWs, elided}
\\
\bob &=& \modelcomment{order imposed by barriers, elided}
\\
\dashboxed[63mm]{}\tfence &=& \po\cap((\neg\stxn\semi\stxn)\cup(\stxn\semi\neg\stxn))
\\
\ob &=& \come \cup \dob \cup \aob \cup \bob \dashboxed{{}\cup \tfence}
\\
\axiom{\AtomicRMW}{\isempty(\rmw \cap (\fre\semi \coe))} \\
\newaxiom{\StrongIsol}{\acyclic(\stronglift(\com,\stxn))} \\
\newaxiom{\TxnOrder}{\acyclic(\stronglift(\ob,\stxn))} \\
\newaxiom{\TxnCancelsRMW}{\isempty(\rmw \cap \tfence^*)} \\
\end{axiomatisation}
\caption{ARMv8 consistency axioms~\cite{arm17,deacon17}, with our TM additions \dashboxed{\text{highlighted}}, and some details elided for brevity.}
\label{fig:axioms_arm}
\end{figure}

An ARMv8 execution is consistent if it satisfies all of the axioms in Fig.~\ref{fig:axioms_arm} (ignoring the highlighted regions).
We have seen the \Coherence{} and \AtomicRMW{} axioms already.
The ordered-before relation ($\ob$) plays the same role as happens-before in x86: it imposes the event-ordering constraints upon which all threads must agree, and must be free from cycles (\Order).
These constraints arise from communication ($\come$), dependencies ($\dob$), atomic RMWs ($\aob$), and barriers ($\bob$). 

\subsection{Adding Transactions}

The ARMv8 architecture does not support TM, so the extensions proposed below (highlighted in Fig.~\ref{fig:axioms_arm}) are unofficial.
Nonetheless, the extensions we give are based upon a proposal currently being considered within ARM Research and upon extensive conversations with ARM architects.

\begin{itemize}

\item \StrongIsol{} is a natural choice for hardware TM.

\item As in x86 and Power, we place implicit fences ($\tfence$) at the boundaries of successful transactions.

\item We bring the \TxnOrder{} axiom from x86 and Power to forbid $\ob$-cycles among transactions.

\item Like Power, ARMv8 has exclusive instructions, so it inherits the \TxnCancelsRMW{} axiom to ensure the failure of RMWs that straddle a transaction boundary.

\end{itemize}

\subsection{Empirical Testing}
\label{sec:arm:testing}

ARM hardware does not support TM so we cannot test our model as we did for x86 and Power.
However, we generated the {\bf Forbid} and {\bf Allow} suites anyway, and gave them to ARM architects.
They were able to use these to reveal a bug (specifically, a violation of the \TxnOrder{} axiom) in a register-transfer level (RTL) prototype implementation.

\newcommand\satxn{\mathit{stxn}_{\rm at}}

\newcommand\HbCom{\ax{HbCom}}
\newcommand\NoThinAir{\ax{NoThinAir}}
\newcommand\SeqCst{\ax{SeqCst}}
\newcommand\NoRace{\ax{NoRace}}

\section{Transactions in C++}
\label{sec:cpp}

We now turn our attention from hardware to software. 
TM is supported in C++ via an ISO technical specification that has been under development by the C++ TM Study Group since 2012~\cite{c++tm15, shpeisman+09}. 
In this section, we formalise how the proposed TM extensions interact with the existing C++ memory model, and detail a possible simplification to the specification. 

C++ TM offers two main types of transactions: \emph{relaxed transactions} (written \texttt{synchronized\{...\}}) can contain arbitrary code, but only promise weak isolation, while \emph{atomic transactions} (written \texttt{atomic\{...\}}) promise strong isolation but cannot contain certain operations, such as atomic operations~\cite[\S8.4.4]{c++tm15}.
Some atomic transactions can be aborted by the programmer, but we do not support these in this paper.

\subsection{Background: the C++ Memory Model}
\label{sec:cpp:background}

\newcommand\sw{\mathit{sw}}
\newcommand\cnf{\mathit{cnf}}
\newcommand\psc{\mathit{psc}}
\newcommand\ecom{\mathit{ecom}}

\newcommand\ACQ{\mathit{Acq}}
\newcommand\Ato{\mathit{Ato}}

\newcommand\tsw{\mathit{tsw}}

Our presentation of the baseline C++ memory model follows \citet{lahav+17}. 
We choose to build on their formalisation because it incorporates a fix that allows correct compilation to Power -- without this, we could not check the compilation of C++ transactions to Power transactions (\S\ref{sec:metatheory:compilation}).

C++ executions identify four additional subsets of events: $\Ato$ contains the events from atomic operations, while $\ACQ$, $\REL$, and $\SC$ contain events from atomic operations that use the acquire, release, and SC consistency modes~\cite[\S29.3]{c++11}.

\begin{figure}
\centering
\begin{axiomatisation}
\axiom{\HbCom}{\irreflexive(\hb \semi \com^*)} \\
\where \sw &=& \modelcomment{synchronises-with, elided} \\
\dashboxed[32mm]{}\ecom &=& \com \cup (\co\semi\rf) \\
\dashboxed[38mm]{}\tsw &=& \weaklift(\ecom,\stxn)\\
\hb &=& (\sw \cup \dashboxed{\tsw \cup {}}  \po)^+ \\
\axiom{\AtomicRMW}{\isempty(\rmw \cap (\fre\semi \coe))} \\
\axiom{\NoThinAir}{\acyclic(\po \cup \rf)} \\
\axiom{\SeqCst}{\acyclic(\psc)} \\ 
\where \psc &=& \modelcomment{constraints on SC events, elided} \\
\end{axiomatisation}
\vspace*{-3mm}
\begin{axiomatisation}
\axiom{\NoRace}{\isempty(\cnf \setminus \Ato^2 \setminus (\hb \cup \hb^{-1}))} \\
\where \cnf &=& \stack{((W\!\times\!W) \cup (R\!\times\!W) \cup
(W\!\times\!R)) \cap \sloc \setminus \id} \\
\end{axiomatisation}
\caption{C++ consistency and race-freedom
axioms~\cite{lahav+17}, with our TM additions
\dashboxed{\text{highlighted}}, and some details elided.}
\label{fig:axioms_cpp}
\end{figure}

Unlike the architecture-level memory models, the C++ memory model defines \emph{two} predicates on executions (Fig.~\ref{fig:axioms_cpp}). 
The first characterises the \emph{consistent} candidate executions. 
If any consistent execution violates a second \emph{race-freedom} predicate, then the program is completely undefined. 
Otherwise, the allowed executions are the consistent executions.

A C++ execution is consistent if it satisfies all of the consistency axioms given at the top of Fig.~\ref{fig:axioms_cpp} (ignoring highlighted regions for now).
The first, \HbCom{}, governs the happens-before relation, which is constructed from the program order and the \emph{synchronises-with} relation ($\sw$). 
Roughly speaking, an $\sw$ edge is induced when an acquire read observes a release write; but it also handles fences and the `release sequence'~\cite{lahav+17, batty+16}. 
The second axiom is standard for capturing the isolation of RMW operations.
The \NoThinAir{} axiom is \citeauthor{lahav+17}'s solution to C++'s `thin air' problem~\cite{batty+15}.
Finally, \SeqCst{} forbids certain cycles among $\SC$ events; we omit its definition as it does not interact with our TM extensions.

A consistent C++ execution is race-free if it satisfies the \NoRace{} axiom at the bottom of Fig.~\ref{fig:axioms_cpp}, which states that whenever two conflicting ($\cnf$) events in different threads are not both atomic, they must be ordered by happens-before. 

\subsection{Adding Transactions}
\label{sec:cpp:adding_transactions}

The specification for C++ TM makes two amendments to the C++ memory model: one for data races, and one for transactional synchronisation.

\paragraph{Transactions and Data Races}

The definition of a race is unchanged in the presence of TM. In particular, the program
\begin{center}
\begin{tabular}{l||l}
\texttt{atomic\{ x=1; \}} & \texttt{atomic\_store(\&x,2);}
\end{tabular}
\end{center}
is racy -- which is perhaps contrary to the intuition that an atomic transaction with a single non-atomic store should be interchangeable with a non-transactional atomic store.

\begin{Remark}
\label{rem:cpp_abort}
The specification also clarifies that although events in an unsuccessful transaction are unobservable, they can still participate in races. 
This implies that the program
\begin{center}
\begin{tabular}{l||l}
\texttt{atomic\{ x=1; abort(); \}} & \texttt{atomic\_store(\&x,2);}
\end{tabular}
\end{center}
must be considered racy. 
In our formalisation, transactions either succeed (giving rise to an $\stxn$-class) or fail, giving rise to no events (cf.~\S\ref{sec:transactions:executions}). 
This treatment correctly handles races involving unsuccessful transactions, because the race will be detected in the case where the transaction succeeds, but it cannot handle transactions that \emph{never} succeed, such as the one above. 
Therefore, we leave the handling of "abort()" for future work.
\end{Remark}

\paragraph{Transactional Synchronisation}
The second amendment by the C++ TM extension defines when two transactions synchronise~\cite[\S1.10]{c++tm15}.
An execution is deemed consistent only if there is a total order on transactions such that:
\begin{enumerate}

\item this order does not contradict happens-before, and

\item if transaction $T_1$ is ordered before conflicting transaction $T_2$, then the end of $T_1$ synchronises with the start of $T_2$.

\end{enumerate}

\newcommand\txnord{\mathit{to}}

We could incorporate these requirements into the formal model by extending executions with a transaction-ordering relation, $\txnord$, that serialises all the $\stxn$-classes in an order that does not contradict happens-before (point 1), and updating the synchronises-with relation to include events in conflicting transactions that are ordered by $\txnord$ (point 2).

However, this formulation is unsatisfying.
It is awkward that $\txnord$ is used to \emph{define} happens-before but is also forbidden to \emph{contradict} happens-before. 
Moreover, having the consistency predicate involve quantification over all possible transaction serialisations makes simulation expensive~\cite{batty+16}.

Fortunately, we can formulate the C++ TM memory model without relying on a total order over transactions. 
The idea is that if two transactions are in conflict, then their order can be deduced from the existing $\rf$, $\co$, and $\fr$ edges, and if they are not, then there is no need to order them.

In more detail, and with reference to the highlighted parts of Fig.~\ref{fig:axioms_cpp}: observe that whenever two events in an execution conflict ($\cnf$), they must be connected one way or the other by `extended communication' ($\ecom$), which is the communication relation extended with $
\co\semi\rf$ chains. 
That is, $\cnf = \ecom\cup\ecom^{-1}$~[\isabelleqed]. 
We then say that transactions synchronise with ($\tsw$) each other in $\ecom$ order,
and we extend happens-before to include $\tsw$.


By simply extending the definition of $\hb$ like this, we avoid the need for the $\txnord$ relation altogether, and we avoid adding any axioms to the consistency predicate. 
To make our proposal concrete, we provide
\ifdefined\EXTENDED
in \S\ref{sec:cpp_amendment}
\else
in our companion material
\fi
some text that the specification could incorporate (currently under review by the C++ TM Study Group).

\paragraph{Strong Isolation for Atomic Transactions}
The semantics described thus far provides the desired weak-isolating behaviour for \emph{relaxed} transactions; that is, the \WeakIsol{} axiom follows from the other C++ consistency axioms~[\isabelleqed].
However, \emph{atomic} transactions must be strongly isolated.
In fact, atomic transactions enjoy strong isolation simply by being forbidden to contain atomic operations. 
The idea is that for a non-transactional event to observe or affect a transaction's intermediate state, it must conflict with an event in that transaction. 
If this event cannot be atomic, there must be a race. 
Thus, for race-free programs, atomic transactions are guaranteed to be strongly isolated.

To formalise this property, we extend C++ executions with an $\satxn$ relation that identifies a subset of transactions as atomic. It satisfies $\satxn\subseteq\stxn$ and $(\satxn\semi\stxn)\subseteq\satxn$. We then prove the following theorem.

\begin{theorem}[Strong isolation for atomic transactions]
If \ax{NoRace} holds, and atomic transactions contain no atomic operations (i.e., $\domain(\satxn)\cap \Ato = \emptyset$), then \[\acyclic(\stronglift(\com,\satxn)).\]
\end{theorem}
\begin{proof}[Proof sketch]
\renewcommand\qed{\hfill\isabelleqed}
A cycle in $\stronglift(\com,\satxn)$ is either a $\com$-cycle or an $r$-cycle, where $r = \satxn\semi (\com \setminus \satxn)^+ \semi \satxn$. 
From \ax{NoRace}, we have $\com \setminus \Ato^2 \subseteq \hb$. Using this and the expansion $\com^+ = \ecom \cup (\fr\semi\rf)$ we can obtain $r \subseteq \hb$.
To finish the proof, note that execution well-formedness forbids $\com$-cycles, and that $r$-cycles are forbidden too because they are also $\hb$-cycles, which violate \HbCom{}.
\end{proof}

\paragraph{A Transactional SC-DRF Guarantee}
\label{sec:cpp:tdrf}

A central property of the C++ memory model is its SC-DRF guarantee~\cite{adve+90, boehm+08}: all race-free C++ programs that avoid non-SC atomic operations enjoy SC semantics.
This guarantee can be lifted to a transactional setting~\cite{dalessandro+09, shpeisman+09}: all race-free C++ programs that avoid relaxed transactions and non-SC atomic operations enjoy TSC semantics (cf. \S\ref{sec:transactions:tsc}). 
This is formalised in the following theorem, which we prove
\ifdefined\EXTENDED
in \S\ref{sec:tdrf_proof}.
\else
in our companion material.
\fi

\newcommand\tdrfstatement{
If a C++-consistent execution has
\begin{itemize}
\item no relaxed transactions (i.e. $\stxn = \satxn$),
\item no non-SC atomics (i.e. $\Ato = \SC$), and
\item no data races (i.e. \NoRace{} holds),
\end{itemize}
then it is consistent under TSC.
}
\begin{theorem}[Transactional SC-DRF guarantee]
\label{thm:tdrf}
\tdrfstatement
\end{theorem}

\newcommand\scs{\mathit{scr}}
\newcommand\scst{\mathit{scr}^{\rm t}}

\begin{table}
\caption{Summary of our metatheoretical results. Timings are for a machine with four 16-core Opteron processors and 128GB RAM, using the Plingeling solver~\cite{biere10}. A \greencross{} means the property holds up to the given number of events, a \redtick{} means a counterexample was found, and \timeout{} indicates a timeout.}
\label{tab:metatheory_results}
\centering
\renewcommand\tabcolsep{0.9mm}
\begin{tabular}{lll@{\hspace*{-4mm}}rrc}
\toprule
\bf Property & \bf \S & \bf Target & \bf Events & \bf ~~Time & \bf C'ex? \\
\midrule
Monotonicity & \ref{sec:metatheory:monotonicity} 
 & x86       & 6 & 20m 
                            & \greencross \\
&& Power     & 2 & <1s      & \redtick    \\
&& ARMv8     & 2 & <1s      & \redtick    \\
&& C++       & 6 & 91h 
                            & \greencross \\
Compilation & \ref{sec:metatheory:compilation} 
 & C++/x86   & 6 & 14h 
                            & \greencross \\
&& C++/Power & 6 & 16h 
                            & \greencross \\
&& C++/ARMv8 & 6 & 20h 
                            & \greencross \\
Lock elision & \ref{sec:metatheory:hle} 
 & x86       & 8 & >48h     & \timeout    \\
&& Power     & 9 & >48h     & \timeout    \\
&& ARMv8     & 7 & 63s      & \redtick    \\
&& ARMv8 (fixed) & 8 & >48h & \timeout    \\
\bottomrule
\end{tabular}
\end{table}

\section{Metatheory}
\label{sec:metatheory}

We now study several metatheoretical properties of our proposed models.
For instance, one straightforward but important property, which follows immediately from the model definitions, is that our TM models give the same semantics to transaction-free programs as the original models~[\isabelleqed].
In this section, we use \Memalloy{} to check some more interesting properties of our models, as summarised in Tab.~\ref{tab:metatheory_results}.


\subsection{Monotonicity}
\label{sec:metatheory:monotonicity}

\newcommand\piedge{\begin{tikzpicture}[baseline=0]
\draw[edgepi] (0,0.09) to (0.4,0.09);
\end{tikzpicture}}

We check that adding $\stxn$-edges can never make an inconsistent execution consistent. This implies that all of the following program transformations are sound: introducing a transaction (e.g., \begin{tikzpicture}[baseline=(a1.base)]
\node[inner sep=0] (a1) at (0,0) {$\bullet$};
\end{tikzpicture}\,\piedge\,\begin{tikzpicture}[baseline=(a2.base)]
\node(a2) at (0.9,0) {$\bullet$};
\node[stxn, fit=(a2), inner sep=0] {};
\end{tikzpicture}), enlarging a transaction (e.g., \begin{tikzpicture}[baseline=(a1.base)]
\node(a1) at (0,0) {$\bullet$};
\node[inner sep=0](b1) at (0.35,0) {$\bullet$};
\node[stxn, fit=(a1), inner sep=0] {};
\end{tikzpicture}\,\piedge\,\begin{tikzpicture}[baseline=(a2.base)]
\node(a2) at (1.2,0) {$\bullet$};
\node(b2) at (1.55,0) {$\bullet$};
\node[stxn, fit=(a2)(b2), inner sep=0] {};
\end{tikzpicture}), and coalescing two consecutive transactions (e.g., \begin{tikzpicture}[baseline=(a1.base)]
\node(a1) at (0,0) {$\bullet$};
\node(b1) at (0.5,0) {$\bullet$};
\node[stxn, fit=(a1), inner sep=0] {};
\node[stxn, fit=(b1), inner sep=0] {};
\end{tikzpicture}\,\piedge\,\begin{tikzpicture}[baseline=(a2.base)]
\node(a2) at (1.35,0) {$\bullet$};
\node(b2) at (1.7,0) {$\bullet$};
\node[stxn, fit=(a2)(b2), inner sep=0] {};
\end{tikzpicture}).

\Memalloy{} confirmed that the transactional x86 and C++ models enjoy this monotonicity property for all executions with up to 6 events. 
For Power and ARMv8, it found the following counterexample:
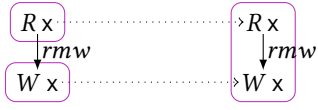
\begin{center}
\begin{tikzpicture}[inner sep=1pt, yscale=-1]
\def\srcposx{0}
\def\srcposy{0}
\def\vspacing{0.8}
\def\hspacing{1.8}
\node (a1) at (\srcposx,\srcposy+0*\vspacing) {$\evR{}{"x"}{0}$};
\node (a2) at (\srcposx,\srcposy+1*\vspacing) {$\evW{}{"x"}{1}$};

\node[stxn, fit=(a1)] {};
\node[stxn, fit=(a2)] {};

\def\tgtposx{3}
\def\tgtposy{0}
\def\vspacing{0.8}
\node (b1) at (\tgtposx,\tgtposy+0*\vspacing) {$\evR{}{"x"}{0}$};
\node (b2) at (\tgtposx,\tgtposy+1*\vspacing) {$\evW{}{"x"}{1}$};

\node[stxn, fit=(b1)(b2)] {};

\draw[edgepi] (a1) to (b1);
\draw[edgepi] (a2) to (b2);

\draw[edgepo] (a1) to [auto] node {$\rmw$} (a2);
\draw[edgepo] (b1) to [auto] node {$\rmw$} (b2);
\end{tikzpicture}
\end{center}
The left execution is inconsistent in both models because of the \TxnCancelsRMW{} axiom: that a store-exclusive separated from its corresponding load-exclusive by a transaction boundary always fails. 
The right execution, however, is consistent.
This counterexample implies that techniques that involve transaction coalescing~\cite{chung+08, stipic+13} must be applied with care in the presence of RMWs.

\subsection{Mapping C++ Transactions to Hardware}
\label{sec:metatheory:compilation}

We check that it is sound to compile C++ transactions to x86, Power, and ARMv8 transactions.
A realistic compiler would be more complex -- perhaps including fallback options for when hardware transactions fail -- but our direct mapping is nonetheless instructive for comparing the guarantees given to transactions in software and in hardware.

Specifically, we use \Memalloy{} to search for a pair of executions, $X$ and $Y$, such that $X$ is an inconsistent C++ execution, $Y$ is a consistent x86/Power/ARMv8 execution, and $X$ is related to $Y$ via the relevant compilation mapping, encoded in the relation $\pi$. 
Such a pair would demonstrate that the compilation mapping is invalid. 
\citet{wickerson+17} have encoded non-transactional compilation mappings; we only need to extend them to handle transactions, which we do by requiring $\pi$ to preserve all $\stxn$-edges:
\begin{eqnarray*}
\stxn_Y &=& \pi^{-1}\semi \stxn_X \semi \pi.
\end{eqnarray*}

\Memalloy{} confirmed that compilation to x86, Power, and ARMv8 is sound for all C++ executions with up to 6 events.

\subsection{Checking Lock Elision}
\label{sec:metatheory:hle}

We now check the soundness of lock elision in x86, Power, and ARMv8 using the technique proposed in \S\ref{sec:methodology:libabs}.

First, we extend executions with four new event types:
\begin{itemize}
\item $\evL$, the "lock()" calls that will be implemented by acquiring the lock in the ordinary fashion,
\item $\evU$, the corresponding "unlock()" calls,
\item $\evLt$, the "lock()" calls that will be transactionalised, and
\item $\evUt$, the corresponding "unlock()" calls.
\end{itemize}
When generating candidate executions from programs, we assume that each "lock()"/"unlock()" pair gives rise to a $L$-$U$ pair or a $\evLt$-$\evUt$ pair. 
(Distinguishing these two modes at the execution level eases the definition of the mapping relation.)
We obtain from these lock/unlock events a derived $\scs$ relation that forms an equivalence class among all the events in the same CR.
Similarly, $\scst$ is a subrelation of $\scs$ that comprises just those CRs that are to be transactionalised.

Second, we extend execution well-formedness so that every $\evL$ event must be followed by a $\evU$ event without an intervening $\evLt$ or $\evUt$, and so on.

Third, the consistency predicates from Figs.~\ref{fig:axioms_x86}, \ref{fig:axioms_power}, and~\ref{fig:axioms_arm} are extended with the following axiom that forces the serialisability of CRs.
\begin{align}
\tag{\ax{CROrder}}
\acyclic(\weaklift(\po\cup\com, \scs))
\end{align}

\begin{table}
\caption{Key constraints on $\pi$ for defining lock elision}
\label{tab:hle_mapping}
\centering
\renewcommand{\tabcolsep}{0.3mm}
\newcommand\xstart{6mm}
\newcommand\xend{5mm}
\newcommand\pstart{6mm}
\newcommand\pend{14mm}
\newcommand\astart{6mm}
\newcommand\aend{11mm}
\begin{tabular*}{\linewidth}{ccccc}
\toprule
\multirow{2}{*}{
\begin{tabular}[b]{@{}c@{}}\textbf{Source}\\\textbf{event}, $e$\end{tabular}}
    & \multicolumn{4}{c}{\textbf{Target event(s)}, $\pi(e)$} \\
\cmidrule(l){2-5}
    & x86 & Power & ARMv8 & ARMv8 (fixed)\\
\midrule
\arrayrulecolor{black!50}
$\evL$ & 
\begin{tikzpicture}[inner sep=1pt,yscale=-1, baseline=(b.base), trim left=-\xstart, trim right=\xend]
\node (c) at (0,-0.6) {$\evR{}{}{}$};
\node (a) at (0,0) {$\evR{}{}{}$};
\node (b) at (0,0.7) {$\evW{}{}{}$};
\draw[edgepo] (c) to[auto] node{$\ctrl$} (a);
\draw[edgepo] (a) to[auto] node{$\rmw$} (b);
\end{tikzpicture} & 
\begin{tikzpicture}[inner sep=1pt,yscale=-1, baseline=(b.base), trim left=-\pstart, trim right=\pend]
\node (a) at (0,0) {$\evR{}{}{}$};
\node (b) at (0,0.7) {$\evW{}{}{}$};
\node (c) at (0,1.4) {"isync"};
\draw[edgepo] (a) to[auto] node{$\rmw,\ctrl$} (b);
\draw[edgepo] (b) to[auto] node{$\ctrl$} (c);
\end{tikzpicture} & 
\begin{tikzpicture}[inner sep=1pt,yscale=-1, baseline=(b.base), trim left=-\astart, trim right=\aend]
\node (a) at (0,0) {\vphantom{$W$}\smash{$\evR{\Acq}{}{}$}};
\node (b) at (0,0.7) {$\evW{}{}{}$};
\draw[edgepo] (a) to[auto] node{$\rmw,\ctrl$} (b);
\end{tikzpicture} & 
\begin{tikzpicture}[inner sep=1pt,yscale=-1, baseline=(b.base), trim left=-\astart, trim right=\aend]
\node (a) at (0,0) {\vphantom{$W$}\smash{$\evR{\Acq}{}{}$}};
\node (b) at (0,0.7) {$\evW{}{}{}$};
\node (c) at (0,1.4) {"dmb"};
\draw[edgepo] (a) to[auto] node{$\rmw,\ctrl$} (b);
\draw[edgepo] (b) to[auto] node{$\po$} (c);
\end{tikzpicture}
\\ \midrule
$\evU$ & 
\begin{tikzpicture}[inner sep=1pt,yscale=-1, trim left=-\xstart, trim right=\xend]
\node (a) at (0,0) {$\evW{}{}{}$};
\end{tikzpicture} & 
\begin{tikzpicture}[inner sep=1pt,yscale=-1, baseline=(b.base), trim left=-\pstart, trim right=\pend]
\node (a) at (0,0) {"sync"};
\node (b) at (0,0.7) {$\evW{}{}{}$};
\draw[edgepo] (a) to[auto] node{$\po$} (b);
\end{tikzpicture} & 
\begin{tikzpicture}[inner sep=1pt,yscale=-1, trim left=-\astart, trim right=\aend]
\node (a) at (0,0) {\vphantom{$W$}\smash{$\evW{\Rel}{}{}$}};
\end{tikzpicture} & 
\begin{tikzpicture}[inner sep=1pt,yscale=-1, trim left=-\astart, trim right=\aend]
\node (a) at (0,0) {\vphantom{$W$}\smash{$\evW{\Rel}{}{}$}};
\end{tikzpicture}
\\ \midrule
$\evLt$ &
\begin{tikzpicture}[inner sep=1pt,yscale=-1, baseline=(b.base), trim left=-\xstart, trim right=\xend]
\node (b) at (0,0) {$\evR{}{}{}$};
\end{tikzpicture} &
\begin{tikzpicture}[inner sep=1pt,yscale=-1, baseline=(b.base), trim left=-\pstart, trim right=\pend]
\node (b) at (0,0) {$\evR{}{}{}$};
\end{tikzpicture} &
\begin{tikzpicture}[inner sep=1pt,yscale=-1, baseline=(b.base), trim left=-\astart, trim right=\aend]
\node (b) at (0,0) {$\evR{}{}{}$};
\end{tikzpicture} &
\begin{tikzpicture}[inner sep=1pt,yscale=-1, baseline=(b.base), trim left=-\astart, trim right=\aend]
\node (b) at (0,0) {$\evR{}{}{}$};
\end{tikzpicture}
\\ \midrule
$\evUt$ & - & - & - & - \\
\arrayrulecolor{black}
\bottomrule
\end{tabular*}

\begin{axiomatisationWithoutBox}
\header{Moreover:} \\
\axiom{LockVar}{\sloc_Y ~~=~~ I^2 \cup ((\neg I)^2 \cap (\pi^{-1}\semi \sloc_X\semi\pi))} \\
\where I = \pi(\evL\cup\evU\cup\evLt\cup\evUt) 
\\
\axiom{TxnIntro}{\scst\setminus (\neg\evUt)^2 ~~=~~ \pi\semi\stxn_Y\semi\pi^{-1}} 
\\
\axiom{TxnReadsLockFree}{\isempty([\evL]\semi \pi\semi\rf\semi\pi^{-1}\semi[\evLt])}
\\
\bottomrule
\end{axiomatisationWithoutBox}
\end{table}


Finally, we define a mapping $\pi$ from the events of an `abstract' execution $X$ to those of a `concrete' execution $Y$, that captures the implementation of lock elision.
Table~\ref{tab:hle_mapping} sketches the main constraints we impose on $\pi$ so that it captures lock elision for x86, Power, and ARMv8. 
It preserves all the execution structure except for lock/unlock events.
The \ax{LockVar} constraint imposes that all the reads/writes in $Y$ that are introduced by the mapping (call these $I$) access the same location (i.e., the lock variable) and that this location is not accessed elsewhere in $Y$. 
The \ax{TxnIntro} constraint imposes that events in the same transactionalised CR in $X$ become events in the same transaction in $Y$. 
The $\evL$ and $\evU$ events are mapped to sequences of events that represent the recommended spinlock implementation for each architecture. 
Each $\evL$ event maps to a successful RMW on the lock variable, which in ARMv8 is an acquire-RMW~\cite[\S K.9.3.1]{arm17}, in Power is followed by a control dependency\footnote{In Power, $\ctrl$ edges can begin at a store-exclusive~\cite{sarkar+12}.} and an "isync"~\cite[\S B.2.1.1]{power30}, and in x86 is preceded by an additional read (the `test-and-test-and-set' idiom)~\cite[\S8.10.6.1]{intel17}. 
Each $\evU$ event maps to a write on the lock variable, which in ARMv8 is a release-write~\cite[\S K.9.3.2]{arm17}, and in Power is preceded by a "sync"~\cite[\S B.2.2.1]{power30}. 
Each $\evLt$ event maps to a read of the lock variable. 
This read does not observe a write from an $\evL$ event (\ax{TxnReadsLockFree}), to ensure that it sees the lock as free.
Finally, $\evUt$ events vanish (because we do not have explicit events for beginning/ending transactions).

\begin{figure}
\centering
\vspace*{2mm}
\begin{tikzpicture}[inner sep=1pt, yscale=-1]
\def\srcposx{0}
\def\srcposy{0}
\def\vspacing{0.8}
\def\hspacing{1.8}
\node[colorlock] (a1) at (\srcposx,\srcposy+1*\vspacing) {$\evL$};
\node (a2) at (\srcposx,\srcposy+2*\vspacing) {$\evR{}{"x"}{0}$};
\node (a3) at (\srcposx,\srcposy+3*\vspacing) {$\evW{}{"x"}{2}$};
\node[colorlock] (a4) at (\srcposx,\srcposy+4*\vspacing) {$\evU$};
\node[colorlock] (b1) at (\srcposx+\hspacing,\srcposy+1*\vspacing) {$\evLt$};
\node (b2) at (\srcposx+\hspacing,\srcposy+2*\vspacing) {$\evW{}{"x"}{1}$};
\node[colorlock] (b3) at (\srcposx+\hspacing,\srcposy+3*\vspacing) {$\evUt$};

\def\tgtposx{5}
\def\tgtposy{1}
\def\vspacing{0.8}
\node[colorlock] (c1) at (\tgtposx,\tgtposy+0*\vspacing) {$\evR{\smash{\Acq}}{"m"}{0}$};
\node[colorlock] (c2) at (\tgtposx,\tgtposy+1*\vspacing) {$\evW{}{"m"}{1}$};
\node (c3) at (\tgtposx,\tgtposy+2*\vspacing) {$\evR{}{"x"}{0}$};
\node (c4) at (\tgtposx,\tgtposy+3*\vspacing) {$\evW{}{"x"}{2}$};
\node[colorlock] (c5) at (\tgtposx,\tgtposy+4*\vspacing) {$\evW{\Rel}{"m"}{0}$};
\node[colorlock] (d1) at (\tgtposx+\hspacing,\tgtposy+1*\vspacing) {$\evR{}{"m"}{0}$};
\node (d2) at (\tgtposx+\hspacing,\tgtposy+2*\vspacing) {$\evW{}{"x"}{1}$};

\node[stxn, fit=(d1)(d2)] {};

\draw[edgepi,overlay] (a1) to[bend right=15] (c1);
\draw[edgepi] (a1) to (c2);
\draw[edgepi] (a2) to (c3);
\draw[edgepi] (a3) to (c4);
\draw[edgepi] (a4) to (c5);
\draw[edgepi] (b1) to (d1);
\draw[edgepi] (b2) to (d2);

\draw[edgefr] (a2) to [auto] node {$\fr$} (b2);
\draw[edgeco] (b2) to [auto, bend right=20] node {$\co$} (a3);
\draw[edgepo] (a1) to [auto] node {$\po$} (a2);
\draw[edgepo] (a2) to [auto] node {$\po,\data$} (a3);
\draw[edgepo] (a3) to [auto] node {$\po$} (a4);
\draw[edgepo] (b1) to [auto] node {$\po$} (b2);
\draw[edgepo] (b2) to [auto] node {$\po$} (b3);

\draw[edgepo] (c1) to [pos=0.4, auto] node {$\po, \rmw, \ctrl$} (c2);
\draw[edgepo] (c2) to [auto] node {$\po$} (c3);
\draw[edgepo] (c3) to [auto] node {$\po,\data$} (c4);
\draw[edgepo] (c4) to [auto] node {$\po$} (c5);
\draw[edgepo] (d1) to [auto] node {$\po$} (d2);
\draw[edgefr] (c3) to [auto] node {$\fr$} (d2);
\draw[edgefr] (d1) to [auto] node {$\fr$} (c2);
\draw[edgeco] (d2) to [auto, bend right=20] node {$\co$} (c4);
\draw[edgeco, overlay] ([xshift=-4mm]c2.south) to [auto, swap, bend left] node {$\co$} ([xshift=-4mm]c5.north);
\draw[edgefr] (c1) to [auto, swap, bend left] node {$\fr$} (c2);

\end{tikzpicture}

\caption{A pair of executions that demonstrates lock elision being unsound in ARMv8}
\label{fig:hle_bug}
\end{figure}
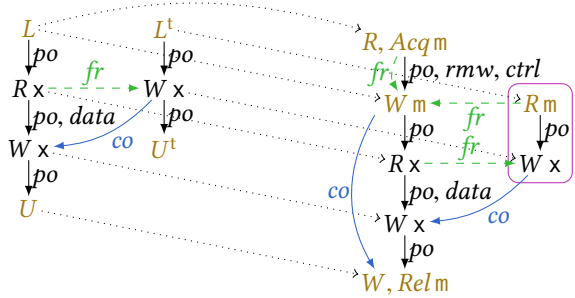

Figure~\ref{fig:hle_bug} shows a pair of ARMv8 executions, $X$ (left) and $Y$ (right), and a $\pi$ relation (dotted arrows), that satisfy all of the constraints above. 
From this example, which was automatically generated using \Memalloy{} in 63 seconds, we manually constructed the pair of litmus tests shown in Example~\ref{ex:hle}.
It thus demonstrates that lock elision is unsound in ARMv8.
This example is actually one of several found by \Memalloy{}; we provide another example
\ifdefined\EXTENDED
in \S\ref{sec:second_lock_elision_example}.
\else
in our companion material.
\fi

We also used \Memalloy{} to check lock elision in x86 and Power, and again in ARMv8 after applying the fix proposed in \S\ref{sec:intro:hlebug} (appending a "DMB" to the "lock()" implementation). 
Given that each architecture implements $\evL$ events with a different number of primitive events (Tab.~\ref{tab:hle_mapping}), we ensured that the event count was large enough in each case to allow examples similar to Fig.~\ref{fig:hle_bug} to be found.
We were unable to find bugs in any of these cases, but \Memalloy{} timed out before it could verify their absence.
As such, we cannot claim lock elision in x86 and Power to be \emph{verified}, but the timeout provides a high degree of confidence that these designs are bug-free up to the given bounds because, as Tab.~\ref{tab:metatheory_results} shows, when counterexamples exist they tend to be found quickly.


\section{Related Work}
\label{sec:related}



In concurrent but independent work, \citet{dongol+18} have also proposed adding TM to the x86, Power, and ARMv8 memory models.
Like us, \citeauthor{dongol+18} build their axioms by lifting relations from events to transactions.
However, their models are significantly weaker than ours, because they capture only the \emph{atomicity} of transactions, not the \emph{ordering} of transactions.
Because of this, their Power model is too weak to validate the natural compiler mapping from C++.
This is demonstrated by the following execution, which is forbidden by C++ (owing to an $\hb$ cycle), but allowed by their Power model (though not actually observable on hardware).
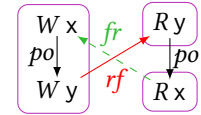
\begin{center}
\begin{tikzpicture}[inner sep=1pt, yscale=-1]
\def\srcposx{0}
\def\srcposy{0}
\def\vspacing{0.9}
\def\hspacing{1.5}

\node (b1) at (\srcposx+\hspacing, \srcposy+1*\vspacing)
  {$\evW{}{"x"}{1}$};
\node (c1) at (\srcposx+\hspacing, \srcposy+2*\vspacing)
  {$\evW{}{"y"}{1}$};  
\node (d1) at (\srcposx+2*\hspacing, \srcposy+1*\vspacing)
  {$\evR{}{"y"}{1}$};
\node (e1) at (\srcposx+2*\hspacing, \srcposy+2*\vspacing)
  {$\evR{}{"x"}{0}$};

\draw[edgepo] (b1) to [auto,swap] node (po) {$\po$} (c1);

\node[stxn, fit=(b1)(c1)(po)] {};
\node[stxn, fit=(d1)] {};
\node[stxn, fit=(e1)] {};

\draw[edgerf] (c1) to [auto,swap,pos=0.3] node {$\rf$} (d1);
\draw[edgepo] (d1) to [auto] node {$\po$} (e1);
\draw[edgefr] (e1) to [auto,swap,pos=0.7] node {$\fr$} (b1);

\end{tikzpicture}
\end{center}
Moreover, unlike our work, \citeauthor{dongol+18}'s models have not been empirically validated -- and nor have earlier models that combine TM and weak memory~\cite{maessen+07, dalessandro+10}.
Nonetheless, our models being stronger than \citeauthor{dongol+18}'s implies that our endeavours are complementary: our experiments validate their models, and their proofs carry over to our models.

\citet{cerone+15} have studied the weak consistency guarantees provided by transactions in database systems.
A key difference is that for \citeauthor{cerone+15}, weak behaviours are attributed to weakly consistent transactions, but in our work, weak behaviours are attributed to weakly consistent non-transactional events surrounding strongly consistent transactions. 
Nonetheless, similar axiomatisations can be used in both settings, and similar weak behaviours can manifest. 

Our models follow the axiomatic style, but there also exist operational memory models for x86~\cite{owens+09}, Power~\cite{sarkar+11}, ARMv8~\cite{pulte+17}, and C++~\cite{nienhuis+16}. It would be interesting to consider how these could be extended to handle TM.

Other architectures and languages that could be targetted by our methodology include RISC-V, which plans to incorporate TM in the future~\cite{riscv17}, and Java.
Indeed, \citet{grossman+06} and \citet{shpeisman+07} identify several tricky corner cases that arise when attempting to extend Java's weak memory model to handle transactions, and our methodology can be seen as a way to automate the generation of these.



Regarding the analysis of programs that \emph{provide} TM, an automatic tool for testing (software) TM implementations above a weak memory model has been developed by \citet{manovit+06}. 
Like us, they use automatically-generated litmus tests to probe the implementations, but where our test suites are constructed to be exhaustive and to contain only `interesting' tests, their tests are randomly generated.
Regarding the analysis of programs that \emph{use} TM, we note that the formulation of the C++ memory model by \citet{lahav+17} leads to an efficient model checker for multithreaded C++~\citet{kokologiannakis+18}.
Since our C++ TM model builds on \citeauthor{lahav+17}'s model, it may be possible to get a model checker for C++ TM similarly.




Regarding tooling for axiomatic memory models in general: our methodology builds on tools due to \citet{wickerson+17} and \citet{lustig+17}, both of which build on Alloy~\cite{jackson12a}. 
Related tools include \Diy{}~\cite{alglave+10}, which generates litmus tests by enumerating relaxations of SC. Compared to \Diy{}, \Memalloy{} is more easily extensible with constructs such as transactions, and only generates the tests needed to validate a model.
\MemSynth{}~\cite{bornholt+17} can synthesise memory models from a corpus of litmus tests and their expected outcomes, though it does not currently handle software models or control dependencies.

\section{Conclusion}

We have extended axiomatic memory models for x86, Power, ARMv8, and C++ to support transactions. 
Using our extensions to \Memalloy{}, we synthesised meaningful sets of litmus tests that precisely capture the subtle interactions between weak memory and transactions. 
These tests allowed us to validate our new models by running them on available hardware, discussing them with architects, and checking them against technical manuals.
We also used \Memalloy{} to check several metatheoretical properties of our models, including the validity of program transformations and compiler mappings, and the correctness -- or lack thereof -- of lock elision.

\specialcomment{acknowledgements}{%
  \begingroup
  \section*{Acknowledgements}
  \phantomsection\addcontentsline{toc}{section}{Acknowledgements}
}{%
  \endgroup
}

\begin{acknowledgements}
We are grateful to Stephan Diestelhorst, Matt Horsnell, and Grigorios Magklis for extensive discussions of TM and the ARM architecture,
to Nizamudheen Ahmed and Vishwanath HV for RTL testing, and
to Peter Sewell for letting us access his Power machine.
We thank the following people for their insightful comments on various drafts of this work:
Mark Batty,
Andrea Cerone,
George Constantinides, 
Stephen Dolan,
Alastair Donaldson,
Brijesh Dongol,
Hugues Evrard,
Shaked Flur,
Graham Hazel,
Radha Jagadeesan,
Jan Ko\'nczak,
Dominic Mulligan,
Christopher Pulte,
Alastair Reid,
James Riely,
the anonymous reviewers,
and our shepherd, Julian Dolby.
This work was supported by
an \grantsponsor{john-icrf}{Imperial College Research Fellowship}{http://www.imperial.ac.uk/research-fellowships/} and
the \grantsponsor{epsrc}{EPSRC}{https://www.epsrc.ac.uk/} (\grantnum{epsrc}{EP/K034448/1}).
\end{acknowledgements}

\balance
\bibliography{bibfile}


\begin{thebibliography}{00}


\ifx \showCODEN    \undefined \def \showCODEN     #1{\unskip}     \fi
\ifx \showDOI      \undefined \def \showDOI       #1{#1}\fi
\ifx \showISBNx    \undefined \def \showISBNx     #1{\unskip}     \fi
\ifx \showISBNxiii \undefined \def \showISBNxiii  #1{\unskip}     \fi
\ifx \showISSN     \undefined \def \showISSN      #1{\unskip}     \fi
\ifx \showLCCN     \undefined \def \showLCCN      #1{\unskip}     \fi
\ifx \shownote     \undefined \def \shownote      #1{#1}          \fi
\ifx \showarticletitle \undefined \def \showarticletitle #1{#1}   \fi
\ifx \showURL      \undefined \def \showURL       {\relax}        \fi
\providecommand\bibfield[2]{#2}
\providecommand\bibinfo[2]{#2}
\providecommand\natexlab[1]{#1}
\providecommand\showeprint[2][]{arXiv:#2}

\bibitem[\protect\citeauthoryear{Adir, Goodman, Hershcovich, Hershkovitz,
  Hickerson, Holtz, Kadry, Koyfman, Ludden, Meissner, Nahir, Pratt, Schliffli,
  St~Onge, Thompto, Tsanko, and Ziv}{Adir et~al\mbox{.}}{2014}]%
        {adir+14}
\bibfield{author}{\bibinfo{person}{Allon Adir}, \bibinfo{person}{Dave Goodman},
  \bibinfo{person}{Daniel Hershcovich}, \bibinfo{person}{Oz Hershkovitz},
  \bibinfo{person}{Bryan Hickerson}, \bibinfo{person}{Karen Holtz},
  \bibinfo{person}{Wisam Kadry}, \bibinfo{person}{Anatoly Koyfman},
  \bibinfo{person}{John Ludden}, \bibinfo{person}{Charles Meissner},
  \bibinfo{person}{Amir Nahir}, \bibinfo{person}{Randall~R. Pratt},
  \bibinfo{person}{Mike Schliffli}, \bibinfo{person}{Brett St~Onge},
  \bibinfo{person}{Brian Thompto}, \bibinfo{person}{Elena Tsanko}, {and}
  \bibinfo{person}{Avi Ziv}.} \bibinfo{year}{2014}\natexlab{}.
\newblock \showarticletitle{Verification of Transactional Memory in {POWER8}}.
  In \bibinfo{booktitle}{{\em Design Automation Conference (DAC)}}.
\newblock
\showDOI{%
\url{https://doi.org/10.1145/2593069.2593241}}


\bibitem[\protect\citeauthoryear{Adve and Hill}{Adve and Hill}{1990}]%
        {adve+90}
\bibfield{author}{\bibinfo{person}{Sarita~V. Adve} {and}
  \bibinfo{person}{Mark~D. Hill}.} \bibinfo{year}{1990}\natexlab{}.
\newblock \showarticletitle{Weak Ordering - A New Definition}. In
  \bibinfo{booktitle}{{\em Int. Symp. on Computer Architecture (ISCA)}}.
\newblock
\showDOI{%
\url{https://doi.org/10.1145/325096.325100}}


\bibitem[\protect\citeauthoryear{Alglave, Maranget, Sarkar, and Sewell}{Alglave
  et~al\mbox{.}}{2010}]%
        {alglave+10}
\bibfield{author}{\bibinfo{person}{Jade Alglave}, \bibinfo{person}{Luc
  Maranget}, \bibinfo{person}{Susmit Sarkar}, {and} \bibinfo{person}{Peter
  Sewell}.} \bibinfo{year}{2010}\natexlab{}.
\newblock \showarticletitle{Fences in Weak Memory Models}. In
  \bibinfo{booktitle}{{\em Computer Aided Verification (CAV)}}.
\newblock
\showDOI{%
\url{https://doi.org/10.1007/978-3-642-14295-6_25}}


\bibitem[\protect\citeauthoryear{Alglave, Maranget, Sarkar, and Sewell}{Alglave
  et~al\mbox{.}}{2011}]%
        {alglave+11a}
\bibfield{author}{\bibinfo{person}{Jade Alglave}, \bibinfo{person}{Luc
  Maranget}, \bibinfo{person}{Susmit Sarkar}, {and} \bibinfo{person}{Peter
  Sewell}.} \bibinfo{year}{2011}\natexlab{}.
\newblock \showarticletitle{Litmus: Running Tests Against Hardware}. In
  \bibinfo{booktitle}{{\em Int. Conf. on Tools and Algorithms for Construction
  and Analysis of Systems (TACAS)}}.
\newblock
\showDOI{%
\url{https://doi.org/10.1007/978-3-642-19835-9_5}}


\bibitem[\protect\citeauthoryear{Alglave, Maranget, and Tautschnig}{Alglave
  et~al\mbox{.}}{2014a}]%
        {alglave+14}
\bibfield{author}{\bibinfo{person}{Jade Alglave}, \bibinfo{person}{Luc
  Maranget}, {and} \bibinfo{person}{Michael Tautschnig}.}
  \bibinfo{year}{2014}\natexlab{a}.
\newblock \showarticletitle{Herding cats: modelling, simulation, testing, and
  data-mining for weak memory}.
\newblock \bibinfo{journal}{{\em ACM Trans. on Programming Languages and
  Systems (TOPLAS)\/}} \bibinfo{volume}{36}, \bibinfo{number}{2}
  (\bibinfo{year}{2014}).
\newblock
\showDOI{%
\url{https://doi.org/10.1145/2627752}}


\bibitem[\protect\citeauthoryear{Alglave, Maranget, and Tautschnig}{Alglave
  et~al\mbox{.}}{2014b}]%
        {alglave+14d}
\bibfield{author}{\bibinfo{person}{Jade Alglave}, \bibinfo{person}{Luc
  Maranget}, {and} \bibinfo{person}{Michael Tautschnig}.}
  \bibinfo{year}{2014}\natexlab{b}.
\newblock \bibinfo{title}{Herding cats: modelling, simulation, testing, and
  data-mining for weak memory (online companion material)}.
  (\bibinfo{year}{2014}).
\newblock
\newblock
\shownote{\url{http://moscova.inria.fr/~maranget/cats/model-power/all.html\#sec4}.}


\bibitem[\protect\citeauthoryear{{ARM}}{{ARM}}{2017}]%
        {arm17}
\bibfield{author}{\bibinfo{person}{{ARM}}.} \bibinfo{year}{2017}\natexlab{}.
\newblock \bibinfo{booktitle}{{\em {ARMv8} Architecture Reference Manual}}.
\newblock
\showURL{%
\url{https://static.docs.arm.com/ddi0487/b/DDI0487B_a_armv8_arm.pdf}}


\bibitem[\protect\citeauthoryear{Batty, Dodds, and Gotsman}{Batty
  et~al\mbox{.}}{2013}]%
        {batty+13}
\bibfield{author}{\bibinfo{person}{Mark Batty}, \bibinfo{person}{Mike Dodds},
  {and} \bibinfo{person}{Alexey Gotsman}.} \bibinfo{year}{2013}\natexlab{}.
\newblock \showarticletitle{Library Abstraction for {C/C++} Concurrency}. In
  \bibinfo{booktitle}{{\em ACM Symp. on Principles of Programming Languages
  (POPL)}}.
\newblock
\showDOI{%
\url{https://doi.org/10.1145/2429069.2429099}}


\bibitem[\protect\citeauthoryear{Batty, Donaldson, and Wickerson}{Batty
  et~al\mbox{.}}{2016}]%
        {batty+16}
\bibfield{author}{\bibinfo{person}{Mark Batty}, \bibinfo{person}{Alastair~F.
  Donaldson}, {and} \bibinfo{person}{John Wickerson}.}
  \bibinfo{year}{2016}\natexlab{}.
\newblock \showarticletitle{Overhauling {SC} atomics in {C11} and {OpenCL}}. In
  \bibinfo{booktitle}{{\em ACM Symp. on Principles of Programming Languages
  (POPL)}}.
\newblock
\showDOI{%
\url{https://doi.org/10.1145/2914770.2837637}}


\bibitem[\protect\citeauthoryear{Batty, Memarian, Nienhuis, Pichon-Pharabod,
  and Sewell}{Batty et~al\mbox{.}}{2015}]%
        {batty+15}
\bibfield{author}{\bibinfo{person}{Mark Batty}, \bibinfo{person}{Kayvan
  Memarian}, \bibinfo{person}{Kyndylan Nienhuis}, \bibinfo{person}{Jean
  Pichon-Pharabod}, {and} \bibinfo{person}{Peter Sewell}.}
  \bibinfo{year}{2015}\natexlab{}.
\newblock \showarticletitle{The Problem of Programming Language Concurrency
  Semantics}. In \bibinfo{booktitle}{{\em Europ. Symp. on Programming (ESOP)}}.
\newblock
\showDOI{%
\url{https://doi.org/10.1007/978-3-662-46669-8_12}}


\bibitem[\protect\citeauthoryear{Biere}{Biere}{2010}]%
        {biere10}
\bibfield{author}{\bibinfo{person}{Armin Biere}.}
  \bibinfo{year}{2010}\natexlab{}.
\newblock \bibinfo{booktitle}{{\em {Lingeling}, {Plingeling}, {PicoSAT} and
  {PrecoSAT} at {SAT Race} 2010}}.
\newblock \bibinfo{type}{{T}echnical {R}eport} 10/1.
  \bibinfo{institution}{Institute for Formal Models and Verification, Johannes
  Kepler University}.
\newblock
\showURL{%
\url{http://fmv.jku.at/papers/Biere-FMV-TR-10-1.pdf}}


\bibitem[\protect\citeauthoryear{Blundell, Lewis, and Martin}{Blundell
  et~al\mbox{.}}{2006}]%
        {blundell+06}
\bibfield{author}{\bibinfo{person}{Colin Blundell}, \bibinfo{person}{E.~C.
  Lewis}, {and} \bibinfo{person}{Milo M.~K. Martin}.}
  \bibinfo{year}{2006}\natexlab{}.
\newblock \showarticletitle{Subtleties of Transactional Memory Atomicity
  Semantics}.
\newblock \bibinfo{journal}{{\em IEEE Computer Architecture Letters\/}}
  \bibinfo{volume}{5}, \bibinfo{number}{2} (\bibinfo{year}{2006}).
\newblock
\showDOI{%
\url{https://doi.org/10.1109/L-CA.2006.18}}


\bibitem[\protect\citeauthoryear{Boehm and Adve}{Boehm and Adve}{2008}]%
        {boehm+08}
\bibfield{author}{\bibinfo{person}{Hans-J. Boehm} {and}
  \bibinfo{person}{Sarita~V. Adve}.} \bibinfo{year}{2008}\natexlab{}.
\newblock \showarticletitle{Foundations of the {C++} Concurrency Memory Model}.
  In \bibinfo{booktitle}{{\em ACM Conf. on Programming Language Design and
  Implementation (PLDI)}}.
\newblock
\showDOI{%
\url{https://doi.org/10.1145/1379022.1375591}}


\bibitem[\protect\citeauthoryear{Bornholt and Torlak}{Bornholt and
  Torlak}{2017}]%
        {bornholt+17}
\bibfield{author}{\bibinfo{person}{James Bornholt} {and} \bibinfo{person}{Emina
  Torlak}.} \bibinfo{year}{2017}\natexlab{}.
\newblock \showarticletitle{Synthesizing Memory Models from Framework Sketches
  and Litmus Tests}. In \bibinfo{booktitle}{{\em ACM Conf. on Programming
  Language Design and Implementation (PLDI)}}.
\newblock
\showDOI{%
\url{https://doi.org/10.1145/3062341.3062353}}


\bibitem[\protect\citeauthoryear{Cain, Frey, Williams, Michael, May, and
  Le}{Cain et~al\mbox{.}}{2013}]%
        {cain+13}
\bibfield{author}{\bibinfo{person}{Harold~W. Cain}, \bibinfo{person}{Brad
  Frey}, \bibinfo{person}{Derek Williams}, \bibinfo{person}{Maged~M. Michael},
  \bibinfo{person}{Cathy May}, {and} \bibinfo{person}{Hung Le}.}
  \bibinfo{year}{2013}\natexlab{}.
\newblock \showarticletitle{Robust Architectural Support for Transactional
  Memory in the Power Architecture}. In \bibinfo{booktitle}{{\em Int. Symp. on
  Computer Architecture (ISCA)}}.
\newblock
\showDOI{%
\url{https://doi.org/10.1145/2485922.2485942}}


\bibitem[\protect\citeauthoryear{Cerone, Bernardi, and Gotsman}{Cerone
  et~al\mbox{.}}{2015}]%
        {cerone+15}
\bibfield{author}{\bibinfo{person}{Andrea Cerone}, \bibinfo{person}{Giovanni
  Bernardi}, {and} \bibinfo{person}{Alexey Gotsman}.}
  \bibinfo{year}{2015}\natexlab{}.
\newblock \showarticletitle{A Framework for Transactional Consistency Models
  with Atomic Visibility}. In \bibinfo{booktitle}{{\em Int. Conf. on
  Concurrency Theory (CONCUR)}}.
\newblock
\showDOI{%
\url{https://doi.org/10.4230/LIPIcs.CONCUR.2015.58}}


\bibitem[\protect\citeauthoryear{Chung, Dalton, Kannan, and Kozyrakis}{Chung
  et~al\mbox{.}}{2008}]%
        {chung+08}
\bibfield{author}{\bibinfo{person}{JaeWoong Chung}, \bibinfo{person}{Michael
  Dalton}, \bibinfo{person}{Hari Kannan}, {and} \bibinfo{person}{Christos
  Kozyrakis}.} \bibinfo{year}{2008}\natexlab{}.
\newblock \showarticletitle{Thread-Safe Dynamic Binary Translation using
  Transactional Memory}. In \bibinfo{booktitle}{{\em Int. Symp. on High
  Performance Computer Architecture (HPCA)}}.
\newblock
\showDOI{%
\url{https://doi.org/10.1109/HPCA.2008.4658646}}


\bibitem[\protect\citeauthoryear{Collier}{Collier}{1992}]%
        {collier92}
\bibfield{author}{\bibinfo{person}{William~W. Collier}.}
  \bibinfo{year}{1992}\natexlab{}.
\newblock \bibinfo{booktitle}{{\em Reasoning about Parallel Architectures}}.
\newblock \bibinfo{publisher}{Prentice Hall}.
\newblock
\showISBNx{0-13-767187-3}


\bibitem[\protect\citeauthoryear{Dalessandro and Scott}{Dalessandro and
  Scott}{2009}]%
        {dalessandro+09}
\bibfield{author}{\bibinfo{person}{Luke Dalessandro} {and}
  \bibinfo{person}{Michael~L. Scott}.} \bibinfo{year}{2009}\natexlab{}.
\newblock \showarticletitle{Strong Isolation is a Weak Idea}. In
  \bibinfo{booktitle}{{\em ACM Workshop on Transactional Computing
  (TRANSACT)}}.
\newblock
\showURL{%
\url{http://transact09.cs.washington.edu/33_paper.pdf}}


\bibitem[\protect\citeauthoryear{Dalessandro, Scott, and Spear}{Dalessandro
  et~al\mbox{.}}{2010}]%
        {dalessandro+10}
\bibfield{author}{\bibinfo{person}{Luke Dalessandro},
  \bibinfo{person}{Michael~L. Scott}, {and} \bibinfo{person}{Michael~F.
  Spear}.} \bibinfo{year}{2010}\natexlab{}.
\newblock \showarticletitle{Transactions as the Foundation of a Memory
  Consistency Model}. In \bibinfo{booktitle}{{\em Int. Conf. on Distributed
  Computing (DISC)}}.
\newblock
\showDOI{%
\url{https://doi.org/10.1007/978-3-642-15763-9_4}}


\bibitem[\protect\citeauthoryear{Deacon}{Deacon}{2016}]%
        {deacon17}
\bibfield{author}{\bibinfo{person}{Will Deacon}.}
  \bibinfo{year}{2016}\natexlab{}.
\newblock \bibinfo{title}{The ARMv8 Application Level Memory Model}.
\newblock
  \bibinfo{howpublished}{\url{https://github.com/herd/herdtools7/blob/master/herd/libdir/aarch64.cat}}.
    (\bibinfo{year}{2016}).
\newblock


\bibitem[\protect\citeauthoryear{Dice, Lev, Moir, and Nussbaum}{Dice
  et~al\mbox{.}}{2009}]%
        {dice+09}
\bibfield{author}{\bibinfo{person}{Dave Dice}, \bibinfo{person}{Yossi Lev},
  \bibinfo{person}{Mark Moir}, {and} \bibinfo{person}{Dan Nussbaum}.}
  \bibinfo{year}{2009}\natexlab{}.
\newblock \showarticletitle{Early Experience with a Commercial Hardware
  Transactional Memory Implementation}. In \bibinfo{booktitle}{{\em Int. Conf.
  on Architectural Support for Programming Languages and Operating Systems
  (ASPLOS)}}.
\newblock
\showDOI{%
\url{https://doi.org/10.1145/2528521.1508263}}


\bibitem[\protect\citeauthoryear{Dongol, Jagadeesan, and Riely}{Dongol
  et~al\mbox{.}}{2018}]%
        {dongol+18}
\bibfield{author}{\bibinfo{person}{Brijesh Dongol}, \bibinfo{person}{Radha
  Jagadeesan}, {and} \bibinfo{person}{James Riely}.}
  \bibinfo{year}{2018}\natexlab{}.
\newblock \showarticletitle{Transactions in Relaxed Memory Architectures}. In
  \bibinfo{booktitle}{{\em ACM Symp. on Principles of Programming Languages
  (POPL)}}.
\newblock
\showDOI{%
\url{https://doi.org/10.1145/3158106}}


\bibitem[\protect\citeauthoryear{Flur, Gray, Pulte, Sarkar, Sezgin, Maranget,
  Deacon, and Sewell}{Flur et~al\mbox{.}}{2016}]%
        {flur+16}
\bibfield{author}{\bibinfo{person}{Shaked Flur}, \bibinfo{person}{Kathryn~E.
  Gray}, \bibinfo{person}{Christopher Pulte}, \bibinfo{person}{Susmit Sarkar},
  \bibinfo{person}{Ali Sezgin}, \bibinfo{person}{Luc Maranget},
  \bibinfo{person}{Will Deacon}, {and} \bibinfo{person}{Peter Sewell}.}
  \bibinfo{year}{2016}\natexlab{}.
\newblock \showarticletitle{Modelling the {ARMv8} Architecture, Operationally:
  {C}oncurrency and {ISA}}. In \bibinfo{booktitle}{{\em ACM Symp. on Principles
  of Programming Languages (POPL)}}.
\newblock
\showDOI{%
\url{https://doi.org/10.1145/2837614.2837615}}


\bibitem[\protect\citeauthoryear{Grossman, Manson, and Pugh}{Grossman
  et~al\mbox{.}}{2006}]%
        {grossman+06}
\bibfield{author}{\bibinfo{person}{Dan Grossman}, \bibinfo{person}{Jeremy
  Manson}, {and} \bibinfo{person}{William Pugh}.}
  \bibinfo{year}{2006}\natexlab{}.
\newblock \showarticletitle{What Do High-Level Memory Models Mean for
  Transactions?}. In \bibinfo{booktitle}{{\em ACM Workshop on Memory Systems
  Performance \& Correctness (MSPC)}}.
\newblock
\showDOI{%
\url{https://doi.org/10.1145/1178597.1178609}}


\bibitem[\protect\citeauthoryear{Hachman}{Hachman}{2014}]%
        {hachman14}
\bibfield{author}{\bibinfo{person}{Mark Hachman}.}
  \bibinfo{year}{2014}\natexlab{}.
\newblock \showarticletitle{Intel finds specialized {TSX} enterprise bug on
  {Haswell}, {Broadwell} {CPUs}}.
\newblock \bibinfo{journal}{{\em PCWorld\/}} (\bibinfo{date}{August}
  \bibinfo{year}{2014}).
\newblock
\showURL{%
\url{http://www.pcworld.com/article/2464880}}


\bibitem[\protect\citeauthoryear{Harris, Larus, and Rajwar}{Harris
  et~al\mbox{.}}{2010}]%
        {harris+10}
\bibfield{author}{\bibinfo{person}{Tim Harris}, \bibinfo{person}{James Larus},
  {and} \bibinfo{person}{Ravi Rajwar}.} \bibinfo{year}{2010}\natexlab{}.
\newblock \bibinfo{booktitle}{{\em Transactional Memory\/}
  (\bibinfo{edition}{2nd} ed.)}.
\newblock \bibinfo{publisher}{Morgan \& Claypool}.
\newblock
\showDOI{%
\url{https://doi.org/10.2200/S00272ED1V01Y201006CAC011}}


\bibitem[\protect\citeauthoryear{Herlihy and Moss}{Herlihy and Moss}{1993}]%
        {herlihy+93}
\bibfield{author}{\bibinfo{person}{Maurice Herlihy} {and}
  \bibinfo{person}{J.~Eliot~B. Moss}.} \bibinfo{year}{1993}\natexlab{}.
\newblock \showarticletitle{Transactional Memory: Architectural Support for
  Lock-Free Data Structures}. In \bibinfo{booktitle}{{\em Int. Symp. on
  Computer Architecture (ISCA)}}.
\newblock
\showDOI{%
\url{https://doi.org/10.1145/173682.165164}}


\bibitem[\protect\citeauthoryear{IBM}{IBM}{2015}]%
        {power30}
\bibfield{author}{\bibinfo{person}{IBM}.} \bibinfo{year}{2015}\natexlab{}.
\newblock \bibinfo{booktitle}{{\em Power ISA (Version 3.0)}}.
\newblock


\bibitem[\protect\citeauthoryear{{Intel}}{{Intel}}{2017}]%
        {skl105}
\bibfield{author}{\bibinfo{person}{{Intel}}.} \bibinfo{year}{2017}\natexlab{}.
\newblock \bibinfo{title}{6th Generation Intel Processor Family: Specification
  Update}.
\newblock   (\bibinfo{date}{June} \bibinfo{year}{2017}).
\newblock
\showURL{%
\url{https://www3.intel.com/content/dam/www/public/us/en/documents/specification-updates/desktop-6th-gen-core-family-spec-update.pdf}}


\bibitem[\protect\citeauthoryear{Intel}{Intel}{2017}]%
        {intel17}
\bibfield{author}{\bibinfo{person}{Intel}.} \bibinfo{year}{2017}\natexlab{}.
\newblock \bibinfo{booktitle}{{\em Intel 64 and IA-32 Architectures: Software
  Developer's Manual}}.
\newblock
\showURL{%
\url{https://software.intel.com/en-us/articles/intel-sdm}}


\bibitem[\protect\citeauthoryear{{Intel Developer Zone}}{{Intel Developer
  Zone}}{2012}]%
        {intel12}
\bibfield{author}{\bibinfo{person}{{Intel Developer Zone}}.}
  \bibinfo{year}{2012}\natexlab{}.
\newblock \bibinfo{title}{Transactional Synchronization in Haswell}.
\newblock   (\bibinfo{date}{February} \bibinfo{year}{2012}).
\newblock
\showURL{%
\url{https://software.intel.com/en-us/blogs/2012/02/07/transactional-synchronization-in-haswell}}


\bibitem[\protect\citeauthoryear{ISO/IEC}{ISO/IEC}{2011}]%
        {c++11}
\bibfield{author}{\bibinfo{person}{ISO/IEC}.} \bibinfo{year}{2011}\natexlab{}.
\newblock \bibinfo{booktitle}{{\em Programming languages -- {C++}}}.
\newblock \bibinfo{publisher}{International standard 14882:2011}.
\newblock


\bibitem[\protect\citeauthoryear{ISO/IEC}{ISO/IEC}{2015}]%
        {c++tm15}
\bibfield{author}{\bibinfo{person}{ISO/IEC}.} \bibinfo{year}{2015}\natexlab{}.
\newblock \bibinfo{booktitle}{{\em Technical Specification for {C++} Extensions
  for Transactional Memory}}.
\newblock \bibinfo{publisher}{Draft technical specification}.
\newblock
\showURL{%
\url{http://www.open-std.org/jtc1/sc22/wg21/docs/papers/2015/n4514.pdf}}


\bibitem[\protect\citeauthoryear{Jackson}{Jackson}{2012}]%
        {jackson12a}
\bibfield{author}{\bibinfo{person}{Daniel Jackson}.}
  \bibinfo{year}{2012}\natexlab{}.
\newblock \bibinfo{booktitle}{{\em Software Abstractions -- Logic, Language,
  and Analysis\/} (\bibinfo{edition}{revised} ed.)}.
\newblock \bibinfo{publisher}{MIT Press}.
\newblock
\showISBNx{0262101149}


\bibitem[\protect\citeauthoryear{Jensen, Hagensen, and Broughton}{Jensen
  et~al\mbox{.}}{1987}]%
        {jensen+87}
\bibfield{author}{\bibinfo{person}{Eric~H. Jensen}, \bibinfo{person}{Gary~W.
  Hagensen}, {and} \bibinfo{person}{Jeffrey~M. Broughton}.}
  \bibinfo{year}{1987}\natexlab{}.
\newblock \bibinfo{booktitle}{{\em A New Approach to Exclusive Data Access in
  Shared Memory Multiprocessors}}.
\newblock \bibinfo{type}{{T}echnical {R}eport} 97663.
  \bibinfo{institution}{Lawrence Livermore National Laboratory}.
\newblock
\showURL{%
\url{https://e-reports-ext.llnl.gov/pdf/212157.pdf}}


\bibitem[\protect\citeauthoryear{Kokologiannakis, Lahav, Sagonas, and
  Vafeiadis}{Kokologiannakis et~al\mbox{.}}{2018}]%
        {kokologiannakis+18}
\bibfield{author}{\bibinfo{person}{Michalis Kokologiannakis},
  \bibinfo{person}{Ori Lahav}, \bibinfo{person}{Konstantinos Sagonas}, {and}
  \bibinfo{person}{Viktor Vafeiadis}.} \bibinfo{year}{2018}\natexlab{}.
\newblock \showarticletitle{Effective Stateless Model Checking for {C/C++}
  Concurrency}. In \bibinfo{booktitle}{{\em ACM Symp. on Principles of
  Programming Languages (POPL)}}.
\newblock
\showDOI{%
\url{https://doi.org/10.1145/3158105}}


\bibitem[\protect\citeauthoryear{Lahav, Vafeiadis, Kang, Hur, and Dreyer}{Lahav
  et~al\mbox{.}}{2017}]%
        {lahav+17}
\bibfield{author}{\bibinfo{person}{Ori Lahav}, \bibinfo{person}{Viktor
  Vafeiadis}, \bibinfo{person}{Jeehoon Kang}, \bibinfo{person}{Chung-Kil Hur},
  {and} \bibinfo{person}{Derek Dreyer}.} \bibinfo{year}{2017}\natexlab{}.
\newblock \showarticletitle{Repairing Sequential Consistency in {C/C++11}}. In
  \bibinfo{booktitle}{{\em ACM Conf. on Programming Language Design and
  Implementation (PLDI)}}.
\newblock
\showDOI{%
\url{https://doi.org/10.1145/3062341.3062352}}


\bibitem[\protect\citeauthoryear{Lamport}{Lamport}{1979}]%
        {lamport79}
\bibfield{author}{\bibinfo{person}{Leslie Lamport}.}
  \bibinfo{year}{1979}\natexlab{}.
\newblock \showarticletitle{How to Make a Multiprocessor Computer That
  Correctly Executes Multiprocess Programs}.
\newblock \bibinfo{journal}{{\it IEEE Trans. Comput.}} \bibinfo{volume}{C-28},
  \bibinfo{number}{9} (\bibinfo{year}{1979}).
\newblock
\showDOI{%
\url{https://doi.org/10.1109/TC.1979.1675439}}


\bibitem[\protect\citeauthoryear{Lustig, Wright, Papakonstantinou, and
  Giroux}{Lustig et~al\mbox{.}}{2017}]%
        {lustig+17}
\bibfield{author}{\bibinfo{person}{Daniel Lustig}, \bibinfo{person}{Andrew
  Wright}, \bibinfo{person}{Alexandros Papakonstantinou}, {and}
  \bibinfo{person}{Olivier Giroux}.} \bibinfo{year}{2017}\natexlab{}.
\newblock \showarticletitle{Automated Synthesis of Comprehensive Memory Model
  Litmus Test Suites}. In \bibinfo{booktitle}{{\em Int. Conf. on Architectural
  Support for Programming Languages and Operating Systems (ASPLOS)}}.
\newblock
\showDOI{%
\url{https://doi.org/10.1145/3037697.3037723}}


\bibitem[\protect\citeauthoryear{Maessen and Arvind}{Maessen and
  Arvind}{2007}]%
        {maessen+07}
\bibfield{author}{\bibinfo{person}{Jan-Willem Maessen} {and}
  \bibinfo{person}{Arvind}.} \bibinfo{year}{2007}\natexlab{}.
\newblock \showarticletitle{Store Atomicity for Transactional Memory}.
\newblock \bibinfo{journal}{{\em Electronic Notes in Theoretical Computer
  Science\/}} \bibinfo{volume}{174}, \bibinfo{number}{9}
  (\bibinfo{year}{2007}).
\newblock
\showDOI{%
\url{https://doi.org/10.1016/j.entcs.2007.04.009}}


\bibitem[\protect\citeauthoryear{Manovit, Hangal, Chafi, McDonald, Kozyrakis,
  and Olukotun}{Manovit et~al\mbox{.}}{2006}]%
        {manovit+06}
\bibfield{author}{\bibinfo{person}{Chaiyasit Manovit},
  \bibinfo{person}{Sudheendra Hangal}, \bibinfo{person}{Hassan Chafi},
  \bibinfo{person}{Austen McDonald}, \bibinfo{person}{Christos Kozyrakis},
  {and} \bibinfo{person}{Kunle Olukotun}.} \bibinfo{year}{2006}\natexlab{}.
\newblock \showarticletitle{Testing Implementations of Transactional Memory}.
  In \bibinfo{booktitle}{{\em Int. Conf. on Parallel Architectures and
  Compilation Techniques (PACT)}}.
\newblock
\showDOI{%
\url{https://doi.org/10.1145/1152154.1152177}}


\bibitem[\protect\citeauthoryear{Nienhuis, Memarian, and Sewell}{Nienhuis
  et~al\mbox{.}}{2016}]%
        {nienhuis+16}
\bibfield{author}{\bibinfo{person}{Kyndylan Nienhuis}, \bibinfo{person}{Kayvan
  Memarian}, {and} \bibinfo{person}{Peter Sewell}.}
  \bibinfo{year}{2016}\natexlab{}.
\newblock \showarticletitle{An Operational Semantics for {C/C++11}
  Concurrency}. In \bibinfo{booktitle}{{\em ACM Int. Conf. on Object-Oriented
  Programming, Systems, Languages, and Applications (OOPSLA)}}.
\newblock
\showDOI{%
\url{https://doi.org/10.1145/3022671.2983997}}


\bibitem[\protect\citeauthoryear{Owens, Sarkar, and Sewell}{Owens
  et~al\mbox{.}}{2009}]%
        {owens+09}
\bibfield{author}{\bibinfo{person}{Scott Owens}, \bibinfo{person}{Susmit
  Sarkar}, {and} \bibinfo{person}{Peter Sewell}.}
  \bibinfo{year}{2009}\natexlab{}.
\newblock \showarticletitle{A Better x86 Memory Model: x86-{TSO}}. In
  \bibinfo{booktitle}{{\em Theorem Proving in Higher Order Logics (TPHOLs)}}.
\newblock
\showDOI{%
\url{https://doi.org/10.1007/978-3-642-03359-9_27}}


\bibitem[\protect\citeauthoryear{Pulte, Flur, Deacon, French, Sarkar, and
  Sewell}{Pulte et~al\mbox{.}}{2018}]%
        {pulte+17}
\bibfield{author}{\bibinfo{person}{Christopher Pulte}, \bibinfo{person}{Shaked
  Flur}, \bibinfo{person}{Will Deacon}, \bibinfo{person}{Jon French},
  \bibinfo{person}{Susmit Sarkar}, {and} \bibinfo{person}{Peter Sewell}.}
  \bibinfo{year}{2018}\natexlab{}.
\newblock \showarticletitle{Simplifying {ARM} Concurrency: Multicopy-atomic
  Axiomatic and Operational Models for {ARMv8}}. In \bibinfo{booktitle}{{\em
  ACM Symp. on Principles of Programming Languages (POPL)}}.
\newblock
\showDOI{%
\url{https://doi.org/10.1145/3158107}}


\bibitem[\protect\citeauthoryear{Rajwar and Goodman}{Rajwar and
  Goodman}{2001}]%
        {rajwar+01}
\bibfield{author}{\bibinfo{person}{Ravi Rajwar} {and} \bibinfo{person}{James~R.
  Goodman}.} \bibinfo{year}{2001}\natexlab{}.
\newblock \showarticletitle{Speculative Lock Elision: Enabling Highly
  Concurrent Multithreaded Execution}. In \bibinfo{booktitle}{{\em Int. Symp.
  on Microarchitecture (MICRO)}}.
\newblock
\showDOI{%
\url{https://doi.org/10.1109/MICRO.2001.991127}}


\bibitem[\protect\citeauthoryear{Sarkar, Memarian, Owens, Batty, Sewell,
  Maranget, Alglave, and Williams}{Sarkar et~al\mbox{.}}{2012}]%
        {sarkar+12}
\bibfield{author}{\bibinfo{person}{Susmit Sarkar}, \bibinfo{person}{Kayvan
  Memarian}, \bibinfo{person}{Scott Owens}, \bibinfo{person}{Mark Batty},
  \bibinfo{person}{Peter Sewell}, \bibinfo{person}{Luc Maranget},
  \bibinfo{person}{Jade Alglave}, {and} \bibinfo{person}{Derek Williams}.}
  \bibinfo{year}{2012}\natexlab{}.
\newblock \showarticletitle{Synchronising {C/C++} and {POWER}}. In
  \bibinfo{booktitle}{{\em ACM Conf. on Programming Language Design and
  Implementation (PLDI)}}.
\newblock
\showDOI{%
\url{https://doi.org/10.1145/2254064.2254102}}


\bibitem[\protect\citeauthoryear{Sarkar, Sewell, Alglave, Maranget, and
  Williams}{Sarkar et~al\mbox{.}}{2011}]%
        {sarkar+11}
\bibfield{author}{\bibinfo{person}{Susmit Sarkar}, \bibinfo{person}{Peter
  Sewell}, \bibinfo{person}{Jade Alglave}, \bibinfo{person}{Luc Maranget},
  {and} \bibinfo{person}{Derek Williams}.} \bibinfo{year}{2011}\natexlab{}.
\newblock \showarticletitle{Understanding {POWER} Multiprocessors}. In
  \bibinfo{booktitle}{{\em ACM Conf. on Programming Language Design and
  Implementation (PLDI)}}.
\newblock
\showDOI{%
\url{https://doi.org/10.1145/1993498.1993520}}


\bibitem[\protect\citeauthoryear{Shasha and Snir}{Shasha and Snir}{1988}]%
        {shasha+88}
\bibfield{author}{\bibinfo{person}{Dennis Shasha} {and} \bibinfo{person}{Marc
  Snir}.} \bibinfo{year}{1988}\natexlab{}.
\newblock \showarticletitle{Efficient and Correct Execution of Parallel
  Programs that Share Memory}.
\newblock \bibinfo{journal}{{\em ACM Trans. on Programming Languages and
  Systems (TOPLAS)\/}} \bibinfo{volume}{10}, \bibinfo{number}{2}
  (\bibinfo{year}{1988}).
\newblock
\showDOI{%
\url{https://doi.org/10.1145/42190.42277}}


\bibitem[\protect\citeauthoryear{Shpeisman, Adl-Tabatabai, Geva, Ni, and
  Welc}{Shpeisman et~al\mbox{.}}{2009}]%
        {shpeisman+09}
\bibfield{author}{\bibinfo{person}{Tatiana Shpeisman},
  \bibinfo{person}{Ali-Reza Adl-Tabatabai}, \bibinfo{person}{Robert Geva},
  \bibinfo{person}{Yang Ni}, {and} \bibinfo{person}{Adam Welc}.}
  \bibinfo{year}{2009}\natexlab{}.
\newblock \showarticletitle{Towards Transactional Memory Semantics for {C++}}.
  In \bibinfo{booktitle}{{\em Symp. on Parallelism in Algorithms and
  Architectures (SPAA)}}.
\newblock
\showDOI{%
\url{https://doi.org/10.1145/1583991.1584012}}


\bibitem[\protect\citeauthoryear{Shpeisman, Menon, Adl-Tabatabai, Balensiefer,
  Grossman, Hudson, Moore, and Saha}{Shpeisman et~al\mbox{.}}{2007}]%
        {shpeisman+07}
\bibfield{author}{\bibinfo{person}{Tatiana Shpeisman}, \bibinfo{person}{Vijay
  Menon}, \bibinfo{person}{Ali-Reza Adl-Tabatabai}, \bibinfo{person}{Steven
  Balensiefer}, \bibinfo{person}{Dan Grossman}, \bibinfo{person}{Richard~L.
  Hudson}, \bibinfo{person}{Katherine~F. Moore}, {and} \bibinfo{person}{Bratin
  Saha}.} \bibinfo{year}{2007}\natexlab{}.
\newblock \showarticletitle{Enforcing Isolation and Ordering in {STM}}. In
  \bibinfo{booktitle}{{\em ACM Conf. on Programming Language Design and
  Implementation (PLDI)}}.
\newblock
\showDOI{%
\url{https://doi.org/10.1145/1273442.1250744}}


\bibitem[\protect\citeauthoryear{Sorensen, Donaldson, Batty, Gopalakrishnan,
  and Rakamari\'{c}}{Sorensen et~al\mbox{.}}{2016}]%
        {sorensen+16a}
\bibfield{author}{\bibinfo{person}{Tyler Sorensen},
  \bibinfo{person}{Alastair~F. Donaldson}, \bibinfo{person}{Mark Batty},
  \bibinfo{person}{Ganesh Gopalakrishnan}, {and} \bibinfo{person}{Zvonimir
  Rakamari\'{c}}.} \bibinfo{year}{2016}\natexlab{}.
\newblock \showarticletitle{Portable Inter-workgroup Barrier Synchronisation
  for {GPUs}}. In \bibinfo{booktitle}{{\em ACM Int. Conf. on Object-Oriented
  Programming, Systems, Languages, and Applications (OOPSLA)}}.
\newblock
\showDOI{%
\url{https://doi.org/10.1145/2983990.2984032}}


\bibitem[\protect\citeauthoryear{Stipi\'{c}, Smiljkovi\'{c}, Unsal, Cristal,
  and Valero}{Stipi\'{c} et~al\mbox{.}}{2013}]%
        {stipic+13}
\bibfield{author}{\bibinfo{person}{Srdan Stipi\'{c}}, \bibinfo{person}{Vesna
  Smiljkovi\'{c}}, \bibinfo{person}{Osman Unsal}, \bibinfo{person}{Adri\'{a}n
  Cristal}, {and} \bibinfo{person}{Mateo Valero}.}
  \bibinfo{year}{2013}\natexlab{}.
\newblock \showarticletitle{Profile-Guided Transaction Coalescing---Lowering
  Transactional Overheads by Merging Transactions}.
\newblock \bibinfo{journal}{{\em ACM Transactions on Architecture and Code
  Optimization\/}} \bibinfo{volume}{10}, \bibinfo{number}{4}
  (\bibinfo{year}{2013}).
\newblock
\showDOI{%
\url{https://doi.org/10.1145/2541228.2555306}}


\bibitem[\protect\citeauthoryear{Waterman and Asanovi\'{c}}{Waterman and
  Asanovi\'{c}}{2017}]%
        {riscv17}
\bibfield{editor}{\bibinfo{person}{Andrew Waterman} {and}
  \bibinfo{person}{Krste Asanovi\'{c}}} (Eds.).
  \bibinfo{year}{2017}\natexlab{}.
\newblock \bibinfo{booktitle}{{\em The {RISC-V} Instruction Set Manual, Volume
  I: User-Level ISA (version 2.2)}}.
\newblock \bibinfo{publisher}{RISC-V Foundation}.
\newblock
\showURL{%
\url{https://riscv.org/specifications/}}


\bibitem[\protect\citeauthoryear{Wickerson, Batty, Sorensen, and
  Constantinides}{Wickerson et~al\mbox{.}}{2017}]%
        {wickerson+17}
\bibfield{author}{\bibinfo{person}{John Wickerson}, \bibinfo{person}{Mark
  Batty}, \bibinfo{person}{Tyler Sorensen}, {and} \bibinfo{person}{George~A.
  Constantinides}.} \bibinfo{year}{2017}\natexlab{}.
\newblock \showarticletitle{Automatically Comparing Memory Consistency Models}.
  In \bibinfo{booktitle}{{\em ACM Symp. on Principles of Programming Languages
  (POPL)}}.
\newblock
\showDOI{%
\url{https://doi.org/10.1145/3009837.3009838}}


\bibitem[\protect\citeauthoryear{Wong}{Wong}{2014}]%
        {wong14}
\bibfield{author}{\bibinfo{person}{Michael Wong}.}
  \bibinfo{year}{2014}\natexlab{}.
\newblock \showarticletitle{Transactional Language Constructs for {C++}}. In
  \bibinfo{booktitle}{{\em C++ Conference (CppCon)}}.
\newblock
\showURL{%
\url{http://bit.ly/2tWk4uz}}


\end{thebibliography}

\ifdefined\EXTENDED
\appendix

\section{Proposed Amendment to the Transactional C++ Specification}
\label{sec:cpp_amendment}

\subsection{Original Text} The original text is as follows~\cite[\S1.10]{c++tm15}:

\begin{enumerate}
\item There is a global total order of execution for all outer blocks. If, in that total order, $T_1$ is ordered before $T_2$,
\begin{itemize}

\item no evaluation in $T_2$ happens before any evaluation in $T_1$ and

\item if $T_1$ and $T_2$ perform conflicting expression evaluations, then the end of $T_1$ synchronizes with the start of $T_2$.

\end{itemize}
\end{enumerate}

\subsection{Proposed Text}  To accommodate the proposal from \S\ref{sec:cpp}, we propose the following replacement text:

\begin{enumerate}

\item An operation $A$ \emph{communicates to} a memory operation $B$ on the same object $M$ if:

\begin{itemize} 

\item $A$ and $B$ are both side effects and $A$ precedes $B$ in the modification order of $M$;

\item $A$ is a side effect, $B$ is a value computation, and the value computed by $B$ is the value stored either by $A$ or by another side effect $C$ that follows $A$ in the modification order of $M$; or

\item $A$ is a value computation, $B$ is a side effect, and $B$ follows in the modification order of $M$ the side effect that stored the value computed by $A$.

\end{itemize}

\item The end of outer transaction $T_1$ synchronises with the start of outer transaction $T_2$ if an operation in $T_1$ communicates to an operation in $T_2$.

\end{enumerate}

\section{A Second Example of Lock Elision Being Unsound in ARMv8}
\label{sec:second_lock_elision_example}

{
\renewcommand\tabcolsep{0.4mm}
\renewcommand\arraystretch{0.9}

\newcommand\xrightbrace[2][1]{%
\def\mylineheight{0.35}%
\raisebox{2.1mm}{%
\smash{%
\begin{tikzpicture}[baseline=(top)]
\coordinate (top) at (0,#1*\mylineheight-0.06);
\coordinate (bottom) at (0,0);
\draw[colorcomment, pen colour={colorcomment}, decoration={calligraphic brace,amplitude=2.1pt}, decorate, line width=1pt] 
  (top) to node[auto, inner sep=0] {~~\begin{tabular}{l}#2\end{tabular}} (bottom);
\end{tikzpicture}}}}

\newcommand\lc[1]{\textcolor{colorlock}{#1}}

Memalloy found a second example of lock elision being unsound in ARMv8.
It is a variant of Example~\ref{ex:hle}, in which rather than an external store interrupting a read-modify-write operation, we have an external load observing an intermediate write.
Specifically, the program below must never satisfy the given postcondition:
\begin{center}
\small
\begin{tabular}{@{}ll@{\hspace{1mm}}||@{\hspace{1mm}}ll@{}}
\hline
\multicolumn{4}{c}{Initially: $"[X0]"=x=0$}                     \\
\hline
\lc{"lock()"}      &           & \lc{"lock()"}     &          \\
    "MOV W5,\#1"   & \xrightbrace[4]{store \\ twice \\ to $x$} 
                               & "LDR W7,[X0]"  & 
                              \xrightbrace[1]{load $x$} \\
    "STR W5,[X0]" &            &   \lc{"unlock()"}  &          \\
    "MOV W5,\#2"  &            &   &          \\
    "STR W5,[X0]" &            &   &          \\
\lc{"unlock()"}   &                                           \\
\hline
\multicolumn{4}{c}{Test: $"W7"=1$}                                 \\
\hline
\end{tabular}
\end{center}
However, if the left thread executes its CR non-transactionally while the right thread uses lock elision, then the resultant program \emph{can} satisfy that postcondition:
\begin{center}
\small
\begin{tabular}{@{}lll@{\hspace{1mm}}||@{\hspace{1mm}}lll@{}}
\hline
\multicolumn{6}{c}{Initially: $"[X0]"=x=0$, $"[X1]"=m=0$}    \\
\hline
\blacknum[2]{1}& \lc{"Loop:"}        & 
 \xrightbrace[6]{atomically \\ update $m$ \\ from 0 \\ to 1} & 
\blacknum[6]{3}& \lc{"TXBEGIN"}   & 
 \xrightbrace{begin txn}                                     \\
               & \lc{"LDAXR W2,[X1]"}&                       & 
               & \lc{"LDR W6,[X1]"}  & 
 \xrightbrace[4]{load $m$ \\ and abort \\ if non-\\ zero}        \\
               & \lc{"CBNZ W2,Loop"} &                       & 
               & \lc{"CBZ W6,L1"}    &                       \\
\blacknum[2]{4}& \lc{"MOV W3,\#1"}   &                       & 
               & \lc{"TXABORT"}      &                       \\
               & \lc{"STXR W4,W3,[X1]"} &                    & 
               & \lc{"L1:"}          &                       \\
               & \lc{"CBNZ W4,Loop"} &                       & 
               &    "LDR W7,[X0]"     & 
 \xrightbrace[1]{load $x$}                           \\
\blacknum[1]{2}   &    "MOV W5,\#1"    & 
 \xrightbrace[4]{store \\ twice \\ to $x$}                         & 
               & \lc{"TXEND"}        &
               \xrightbrace{end txn} \\
               &    "STR W5,[X0]"    &                       & 
               &                     &                       \\
\blacknum[2]{5}&    "MOV W5,\#2"  &                       & 
               &                     & 
                                                     \\
               &    "STR W5,[X0]"    &                       & 
               &                     &                       \\
               &\lc{"STLR WZR,[X1]"} & 
\xrightbrace{$m \leftarrow 0$}                               &      
                                     &                       \\
\hline
\multicolumn{6}{c}{Test: $"W7"=1$}                              \\
\hline
\end{tabular}
\end{center}
As in Example~\ref{ex:hle}, the numbers next to the instructions indicate the order in which they can execute in order to satisfy the postcondition.
This example is interesting because it shows that not only can \emph{loads} be executed speculatively before a successful store-exclusive ("STXR") completes, but so can \emph{stores}.
}

\section{Proof of Theorem~\ref{thm:tdrf}}
\label{sec:tdrf_proof}

{
\renewcommand\thetheorem{\ref{thm:tdrf}}
\begin{theorem}[Transactional SC-DRF guarantee]
\tdrfstatement
\end{theorem}
}
\addtocounter{theorem}{-1}

\newcommand\pocom{\mathit{pocom}}
\newcommand\poloc{\LOC{\po}}
\newcommand\podloc{\DLOC{\po}}
\newcommand\rfsc{\rf^\SC}
\newcommand\polocsc{\po_{\mathrm{loc}}^\SC}

\let\oldxrightarrow\xrightarrow
\renewcommand\xrightarrow[1]{\oldxrightarrow{\mbox{\scriptsize $#1$}}}

In order to prove this theorem, let us assume the three conditions that the theorem assumes, and show that $\acyclic(\po\cup\com)$ and $\acyclic(\stronglift(\po\cup\com, \stxn))$ both hold.
In what follows, we write $\pocom$ for $\po\cup\com$, $\poloc$ for $\po\cap\sloc$, $\podloc$ for $\po\setminus\sloc$, $\polocsc$ for $\poloc\cap\SC^2$, and $\rfsc$ for $\rf \cap \SC^2$.

\begin{lemma}
\label{lem:nonatomic_com_hb}
In race-free executions, communication (other than that between two SC events) induces happens-before; i.e.,
\begin{eqnarray*}
\com \setminus \SC^2 &\subseteq& \hb.
\end{eqnarray*}
\end{lemma}
\begin{proof}
Consider two non-SC events related by $\com$. If they are in the same thread, then \HbCom{} guarantees that they are in $\po$, and hence in $\hb$. If they are in different threads, then \NoRace{} and \HbCom{} guarantee that they are in $\hb$.
\end{proof}

\begin{lemma}
\label{lem:simpler_hb}
In the absence of non-SC atomics, we can simplify happens-before as follows:
\begin{eqnarray*}
\hb &=& (\po \cup \rfsc \cup \tsw)^+.
\end{eqnarray*}
\end{lemma}
\begin{proof}
Recall that happens-before is defined via $\hb = (\po \cup \sw \cup \tsw)^+$, where 
\begin{eqnarray*}
\sw &=& \stack{[\Acq]\semi ([F]\semi\po)^?\semi [W]\semi \poloc^*\semi [W\cap \Ato]\semi {}\\ (\rf\semi\rmw)^*\semi\rf\semi[R\cap \Ato]\semi(po\semi[F])^?\semi[\Rel].}
\end{eqnarray*}
In the absence of non-SC atomics, this simplifies to
\begin{eqnarray*}
\sw &\subseteq& \po^*\semi (\rfsc\semi\po)^*\semi\rfsc\semi\po^*
\end{eqnarray*}
from which the result follows.
\end{proof}

\subsection*{Proof of $\acyclic(\pocom)$}

Suppose toward a contradiction that there is a $\pocom$ cycle. 

If the $\pocom$ cycle passes through no atomics, then each edge of the cycle is either a $\po$ (and hence an $\hb$), or a non-atomic $\com$ edge (and hence an $\hb$ by Lemma~\ref{lem:nonatomic_com_hb}). 
Hence, we have an $\hb$ cycle. 
These are forbidden by \HbCom, so we have a contradiction.
So, we can henceforth assume that the $\pocom$ cycle passes through at least one atomic.
 
Let us divide the cycle at each atomic event, so that each segment of the cycle begins at an atomic, then passes through a chain of zero or more non-atomics, before finishing at an atomic.
Each segment thus takes the form:
\begin{eqnarray}
\mathit{seg} &=& [\SC]\semi\pocom\semi([\neg\SC]\semi\pocom)^*\semi[\SC].
\end{eqnarray}

\begin{lemma}
\label{lem:seg}
$\mathit{seg}\subseteq \hb \cup \co \cup \fr$.
\end{lemma}
\begin{proof}
If the segment is a single $\pocom$ edge between two SC events, then that edge is either in $\po\cup\rfsc$ (and hence in $\hb$) or is in $\co\cup\fr$, as required.
Otherwise, the segment is of the form
\[
\bullet \xrightarrow{\pocom\semi[\neg\SC]}
\bullet \xrightarrow{([\neg\SC]\semi\pocom\semi[\neg\SC])^*}
\bullet \xrightarrow{[\neg\SC]\semi\pocom} \bullet,
\]
and is hence in $\hb$ by Lemma~\ref{lem:nonatomic_com_hb}, as required.
\end{proof}

If there are fewer than two non-$\hb$ segments in our cycle, then we have an immediate violation of \HbCom{}.
So, we can henceforth assume that at least two segments in our cycle are in $(\co\cup\fr)\setminus\hb$.
Moreover, we can assume that these segments are not consecutive, for if they were, the $\co\cup\fr$ edges would collapse together.
We can hence consider our cycle to be built from chains of segments of the following form:
\[
X'\xrightarrow{(\co\cup\fr)\setminus\hb} 
X\xrightarrow{\hb}
Z\xrightarrow{(\co\cup\fr)\setminus\hb} Z'
\]
where $X'$, $X$, $Z$, and $Z'$ are all in $\SC$.
In fact, we can further assume
\begin{equation}
\label{eq:no_hb_chord}
(X',Z)\notin\hb\qquad (X,Z')\notin\hb
\end{equation}
for otherwise the cycle would still collapse into an \HbCom{} violation.

To complete the proof, it suffices to show that $(X,Z)$ is in $\psc^+$, for then the entire cycle becomes a $\psc$ cycle and hence a violation of the \SeqCst{} axiom.
The precise definition of $\psc$ is quite involved~\cite{lahav+17}; here we just rely on
\begin{eqnarray}
[\SC]\semi\podloc\semi\hb\semi\podloc\semi[\SC] &\subseteq& \psc \\{}
[\SC]\semi\pocom\semi[\SC] &\subseteq& \psc
\end{eqnarray}

Consider the $\hb$ chain from $X$ to $Z$. If it passes through no intermediate non-SC events, then it must be a sequence of $\po$ and $\rfsc$ edges (by Lemma~\ref{lem:simpler_hb}).
All of these edges are in $\psc$, so we have $(X,Z)\in\psc^+$ as required.
So, we can henceforth assume that the $\hb$ chain passes through at least one non-SC event.

Let $Y$ be the last non-SC event in the chain (i.e., the closest to $Z$). The following two lemmas analyse the $(Y,Z)$ and $(X,Y)$ parts of the chain separately.

\begin{lemma}
\label{lem:YZ_psc}
The $\hb$ chain from $Y$ to $Z$ contains a $\podloc$ edge that is followed only by $\rfsc$ and $\polocsc$ edges; i.e.,
\[
(Y,Z) \in \hb^*\semi\podloc\semi (\rfsc \cup \polocsc)^*.
\]
\end{lemma}
\begin{proof}
Consider the final part of the $\hb$ chain.
By Lemma~\ref{lem:simpler_hb}, each edge in the chain is a $\tsw$, a $\po$, or an $\rfsc$.
Any suffix of $\rfsc$ or $\polocsc$ edges can be removed, with the remaining part of the chain still ending at an SC event.
The last of the remaining edges cannot be a $\tsw$ because we are at an SC event, and atomics are forbidden inside atomic transactions.
Moreover, if it is a $\podloc$ edge then we are done.
Hence, it remains only to consider the possibility that it is a $\poloc$ edge that begins at a non-SC event.
This would mean we have
\[
Y\xrightarrow{\hb^*}
A\xrightarrow{\poloc\semi (\rfsc \cup \polocsc)^*}
Z\xrightarrow{(\co\cup\fr)\setminus\hb} Z'
\]
for some non-SC event $A$.
Note that $A$ and $Z'$ cannot be related by $\hb$, because $(Z',A)\in\hb$ would be a \HbCom{} violation, and $(A,Z')\in\hb$ would imply $(X,Z')\in\hb$, in contradiction of assumption~\eqref{eq:no_hb_chord}.
Since $Z'$ is a write, it forms a data race with $A$.
This contradicts our \NoRace{} assumption, and hence the proof is complete.
\end{proof}

\begin{lemma}
\label{lem:XY_psc}
The $\hb$ chain from $X$ to $Y$ contains a $\podloc$ edge that is preceded only by $\rfsc$ and $\polocsc$ edges; i.e.,
\[
(X,Y) \in (\rfsc \cup \polocsc)^*\semi \podloc\semi\hb^*.
\]
\end{lemma}
\begin{proof}
Consider the initial part of the $\hb$ chain.
By Lemma~\ref{lem:simpler_hb}, each edge in the chain is a $\tsw$, a $\po$, or an $\rfsc$.
Any prefix of $\rfsc$ or $\polocsc$ edges can be removed, with the remaining part of the chain still beginning at an SC event.
The first of the remaining edges cannot be a $\tsw$ because we are at an SC event, and atomics are forbidden inside atomic transactions.
Moreover, if it is a $\podloc$ edge then we are done.
Hence, it remains only to consider the possibility that it is a $\poloc$ edge ending at a non-SC event, say $A$.
The next edge after $A$ cannot be another $\po$ because two consecutive $\po$ edges would collapse together, and it cannot be an $\rfsc$ because $A$ is non-SC, so it must be a $\tsw$ edge.
This implies that there is another non-SC event between $A$ and $Z$, and hence that $A$ is not $Y$.
We therefore have
\begin{center}
\begin{tikzpicture}
\node (Xp) at (0,0) {$X'$};
\node (X) at (2,0) {$X$};
\node (A) at (5.6,0) {$A$};
\node (B) at (5.6,-1) {$B$};
\node (C) at (6.6,-1) {$C$};
\node (D) at (6.6,0) {$D$};
\node (Y) at (7.6,0) {$Y$};
\draw[->, line width=0.6pt] (Xp) to [auto] node {\scriptsize$(\co\cup\fr)\setminus\hb$} (X);
\draw[->, line width=0.6pt] (X) to [auto] node {\scriptsize$(\rfsc \cup \polocsc)^*\semi\poloc$} (A);
\draw[->, line width=0.6pt] (A) to [auto] node {\scriptsize$\tsw$} (D);
\draw[->, line width=0.6pt] (A) to [auto,swap] node {\scriptsize$\stxn$} (B);
\draw[->, line width=0.6pt] (B) to [auto] node {\scriptsize$\ecom$} (C);
\draw[->, line width=0.6pt] (C) to [auto,swap] node {\scriptsize$\stxn$} (D);
\draw[->, line width=0.6pt] (D) to [auto] node {\scriptsize$\hb^*$} (Y);
\end{tikzpicture}
\end{center}
for some non-SC events $A$, $B$, $C$, and $D$.
Note that $X'$ and $B$ cannot be related by $\hb$, because $(B,X')\in\hb$ would be a \HbCom{} violation, and $(X',B)\in\hb$ would imply $(X',Z)\in\hb$, in contradiction of assumption~\eqref{eq:no_hb_chord}.
This reasoning is also valid with $C$ substituted for $B$.
At least one of $B$ and $C$ must be a write, by the definition of $\ecom$, and hence at least one of them forms a data race with $X'$.
This contradicts our \NoRace{} assumption, and hence the proof is complete.
\end{proof}

By combining Lemmas~\ref{lem:YZ_psc} and~\ref{lem:XY_psc}, we can deduce that $(X,Z)$ is in
\[
(\rfsc \cup \polocsc)^*\semi \podloc\semi\hb^*\semi\podloc\semi (\rfsc \cup \polocsc)^*
\]
and is hence in $\psc^*\semi\psc\semi\psc^*$, as required.

\subsection*{Proof of $\acyclic(\stronglift(\pocom,\stxn))$}

For this proof, we shall rely on one further lemma. 

\begin{lemma} 
\label{lem:stxn_hb}
Happens-before edges can be lifted to relate transactions; i.e.,
\begin{eqnarray*}
\stxn^*\semi (\hb\setminus\stxn)\semi\stxn^* &\subseteq& \hb\setminus\stxn
\end{eqnarray*}
\end{lemma}
\begin{proof}
It suffices to prove the following two inequalities:
\begin{eqnarray}
\label{eq:stxn_hb1}
\stxn\semi (\hb\setminus\stxn) &\subseteq& \hb\setminus\stxn \\
\label{eq:stxn_hb2}
(\hb\setminus\stxn)\semi\stxn &\subseteq& \hb\setminus\stxn.
\end{eqnarray}
We shall provide a proof of \eqref{eq:stxn_hb1}; that of \eqref{eq:stxn_hb2} is similar.
We begin by noting that the following identity of Kleene algebra
\begin{eqnarray*}
(r\cup s)^+ &=& r^+ \cup (r^*\semi s\semi (r\cup s)^*)
\end{eqnarray*}
can be combined with Lemma~\ref{lem:simpler_hb} to give
\begin{eqnarray*}
\hb\setminus\stxn &\subseteq& (\po \cap \stxn)^*\semi (\po\setminus\stxn \cup \rfsc \cup \tsw)\semi \hb^*.
\end{eqnarray*}
Hence, to show \eqref{eq:stxn_hb1}, it suffices to suppose a chain of the form
\[
V \xrightarrow{\stxn} W \xrightarrow{(\po \,\cap\, \stxn)^*}
X \xrightarrow{\po\setminus\stxn \,\cup\, \rfsc \,\cup\, \tsw} 
Y \xrightarrow{\hb^*} Z
\]
and deduce that $(V,Z)\in \hb$.
We do so by a case-split on the $(X,Y)$ edge. 
First, the case $(X,Y)\in\rfsc$ is impossible because $X$ is in an atomic transaction and hence cannot be atomic. 
Second, if $(X,Y)$ is in $\po\setminus\stxn$, then so is $(V,Y)$, and the result follows.
Third, if $(X,Y)$ is in $\tsw$, then so is $(V,Y)$ and the result follows.
\end{proof}

Now, suppose toward a contradiction that there is a cycle in $\stronglift(\pocom,\stxn)$.
Such a cycle is made up of $\stxn$ and $\pocom$ edges, with at least one edge in $\pocom\setminus\stxn$.

Suppose the cycle passes through no atomics.
If the cycle includes no $\stxn$ edges, then we simply have a $\pocom$ cycle, which is a contradiction because we have already proved $\acyclic(\pocom)$. 
If, on the other hand, the cycle includes at least one $\stxn$ edge, then we can divide the cycle at the transactions, so that each segment of the cycle is of the form
\[
\stxn\semi(\pocom\setminus\stxn)^+\semi\stxn
\]
and hence, by Lemma~\ref{lem:nonatomic_com_hb}, of the form $\stxn\semi(\hb\setminus\stxn)^+\semi\stxn$.
Lemma~\ref{lem:stxn_hb} then implies that this cycle is in $\hb$, and is hence forbidden.
So, we can henceforth assume that the cycle passes through at least one atomic.

Let us divide the cycle at each atomic event, so that each segment of the cycle begins at an atomic, then forms a chain of $\pocom$ and $\stxn$ edges through zero or more non-atomics, before finishing at an atomic.
The first and last edge of each segment cannot be an $\stxn$, however, since each segment begins and ends at an atomic.
As a result, Lemma~\ref{lem:seg} still holds under the new definition of a segment, because any $\stxn$ edges in the segment can be folded into the adjacent $\hb$ edges using Lemma~\ref{lem:stxn_hb}.
The rest of the proof is identical to that of $\acyclic(\pocom)$.

\fi

\end{document}